\documentclass[5p]{elsarticle}
\usepackage{graphicx}
\usepackage{lineno}
\modulolinenumbers[5]
\usepackage{hyperref}
\usepackage{color}
\usepackage[figuresleft]{rotating}
\usepackage{caption}
\usepackage{subcaption}
\usepackage{layout}
\usepackage{setspace}
\usepackage{dcolumn}
\usepackage{xtab}
\xentrystretch{-0.15}

\usepackage[T1]{fontenc}
\usepackage{textcomp}		
\usepackage[scaled=0.92]{helvet}
\usepackage{amssymb}
\usepackage{mathtools}
\usepackage{mathrsfs}
\usepackage{amsthm}
\usepackage{tgtermes, tgheros}
\usepackage{array}
\newcolumntype{L}[1]{>{\raggedright\let\newline\\\arraybackslash\hspace{0pt}}m{#1}}
\newcolumntype{C}[1]{>{\centering\let\newline\\\arraybackslash\hspace{0pt}}m{#1}}
\newcolumntype{R}[1]{>{\raggedleft\let\newline\\\arraybackslash\hspace{0pt}}m{#1}}

\makeatletter
\def\ps@pprintTitle{
\let\@oddhead\@empty
 \let\@evenhead\@empty
 \def\@oddfoot{}
 \let\@evenfoot\@oddfoot}
\makeatother
\makeatletter
\let\oldlt\longtable
\let\endoldlt\endlongtable
\def\longtable{\@ifnextchar[\longtable@i \longtable@ii}
\def\longtable@i[#1]{\begin{figure}[t]
\onecolumn
\begin{minipage}{0.5\textwidth}
\oldlt[#1]
}
\def\longtable@ii{\begin{figure}[t]
\onecolumn
\begin{minipage}{0.5\textwidth}
\oldlt
}
\def\endlongtable{\endoldlt
\end{minipage}
\twocolumn
\end{figure}}
\makeatother

\newtheorem*{theorem*}{Theorem}

\newtheorem*{propo*}{Proposition}
\newtheorem*{rem*}{Remark}

\biboptions{authoryear}

\journal{Economic Modelling}
\begin{document}
\begin{frontmatter}
\title{
Multifactor CES General Equilibrium: Models and Applications
}
\author{Jiyoung Kim, Satoshi Nakano and Kazuhiko Nishimura}
\cortext[cor1]{Email (\today)}

\begin{abstract}
Sector specific multifactor CES elasticities of substitution and the corresponding productivity growths are jointly measured by regressing the growths of factor-wise cost shares against the growths of factor prices.
We use linked input-output tables for Japan and the Republic of Korea as the data source for factor price and cost shares in two temporally distant states. 
We then construct a multi-sectoral general equilibrium model using the system of estimated CES unit cost functions, and evaluate the economy-wide propagation of an exogenous productivity stimuli, in terms of welfare.
Further, we examine the differences between models based on \textit{a priori} elasticities such as Leontief and Cobb-Douglas.
\end{abstract}

\begin{keyword}
Multifactor CES \sep Productivity Growth \sep Elasticity of Substitution
\sep General Equilibrium \sep Linked Input-Output Tables
\end{keyword}
\end{frontmatter}

\section{Introduction}
In this study, we measure the multifactor CES elasticity of substitution, jointly with the productivity growth, for multiple industrial sectors, by way of two temporally distant cross-sectional data (i.e., linked input--output tables). 
As we learn the multifactor CES unit cost function, we discover that an industry specific elasticity can be estimated by regressing the growth of factor-wise cost shares against the growth of factor-wise prices.
We also discover that the industry specific productivity growth can be measured via the intercept of the regression line.  
Consequently, we make use of the linked input--output tables in order to observe the cost shares and the price changes spanning over two periods for multiple industrial sectors.

The two-input constant elasticity of substitution (CES) function was first introduced by \citet{acms}, and \citet{uzawa} and \citet{mcfadden_ces} later showed that elasticities were still unique for the case of more than two factor inputs.
Empirical analyses concerning the measurement of CES elasticities \citep[e.g.,][]{McKW, vw, KS_2015} have been 
based upon time series data, while embedding nest structures into the two-input CES framework conforming to the work by \citet{sato}, to handle elasticities between more than two factors of production.
The number of factors and thus of estimable elasticities, can nevertheless be narrowed depending on the availability of time series data. 
Since we are interested in constructing a multisector general equilibrium model that calls for multifactor production functions, we take the advantage of an alternative approach, exploiting cross-sectional data.

When a multisectoral general equilibrium model is established, assessments can be made of the arbitrary productivity shock resulting from technological innovation, in terms of welfare gained.
Previous studies in this regard have assumed a constant and uniform unit elasticity \citep{klein}, or have used empirically estimated elasticities in Translog or multistage (nested) CES functions with a highly aggregated and thus limited number of substitutable factors.
Examples include works by \citet{kyj}, \citet{st}, and \citet{tokutsu}, and many of the works concerning CGE models such as studies by \citet{cge_ere} and \citet{go}.
In contrast, our approach allows us to construct an empirical model of multifactor production with different elasticities of substitution among many (over 350) industrial sectors.
Moreover, this approach allows us to prospectively portray the ex post technological structure following any given exogenous productivity shock and to account for welfare in terms of economy-wide input--output performances.

We measure the welfare changes attributed to the exogenous productivity change by SCS (social cost saved), i.e., the difference in the total primary factor inputs required to net produce a fixed amount of final consumption, given the productivity change.
We find in theory that SCS will be positive (primary factor inputs will always be saved) in every sector if the exogenous productivity is improving, and vice-versa, under the system with uniform CES elasticity less than unity which is inclusive of Cobb--Douglas and Leontief systems.
Hence conversely, such a law may not necessarily hold for the case of CES system with non-uniform elasticities; and this is verified by the empirical analysis of SCS using the estimated multifactor CES system.

The remainder of this paper is organized as follows. 
In the next section, we introduce the basics of multifactor CES elasticity and productivity growth estimation and apply the protocol to linked input--output tables for Japan and the Republic of Korea having sufficient capacity as far as degrees of freedom of the regression.
In Section 3, we replicate the current technological structure as the general equilibrium state of a system of empirically estimated multifactor CES functions; further, we trace out how that structure is transformed by exogenous productivity stimuli. 
Section 4 provides concluding remarks.
 
\section{The Model}
\subsection{Multifactor CES Functions}
A constant-returns multifactor CES production function of an industrial sector (index $j$ omitted) has the following form:
\begin{align*}
y = z f\left( \mathbf{x} \right) = z \left( \sum_{i=0}^n {{\lambda}}_i^{\frac{1}{\sigma}} x_i^{\frac{\sigma - 1}{\sigma}}  \right)^{\frac{\sigma}{\sigma-1}}
\end{align*}
where, $y$ denotes the output and $x_i$ denotes the $i$th factor input.
Here, the share parameters are assumed to maintain ${\lambda}_i > 0$ and $\sum_{i} {\lambda}_i = 1$, while the elasticity of substitution $\sigma \geq 0$ is subject to estimation. 
Also, we are interested in measuring the growth of productivity i.e., $\Delta \ln z$, where $\Delta$ represents temporally distant differences.

Displayed below is the unit cost function compatible with the multifactor CES production function:
\begin{align*}
c=z^{-1}h\left( \mathbf{w} \right)
= \frac{1}{z} \left( \sum_{i=0}^n {\lambda}_i w_i^{\gamma}  \right)^{{1}/{\gamma}}  
\end{align*}
where, $c$ denotes the unit cost of the output, and $w_i$ denotes the $i$th factor price.
Here, we use $\gamma = 1-\sigma$ for convenience.
The cost share of the $i$th input $a_i$ can be determined, in regard to Shephard's lemma, by differentiating the unit cost function:
\begin{align} 
{a}_{i}=\frac{\partial c}{\partial w_i}\frac{w_i}{c} = {\lambda}_i \left( z c / w_i \right)^{-\gamma} 
\label{cs}
\end{align}
By taking the log of both sides, we have
\begin{align*} 
\ln {a}_{i}
= \ln {\lambda}_i -\gamma \ln z + \gamma \ln \left( w_i/c \right) 
\end{align*}
As we observe two temporally distant values for cost shares ($a_i^0$ and $a_i^0$), factor prices ($w_i^0$ and $w_i^1$), and unit costs of outputs as prices ($c^0=w^0$ and $c^1=w^1$) reflecting perfect competition, we find two identities regarding the data:
\begin{align*} 
\ln {a}^0_{i}
&= \ln {\lambda}_i -\gamma \ln z^0 +\gamma \ln \left( w^0_i/w^0 \right) + \epsilon_i^0 
\\
\ln {a}^1_{i}
&= \ln {\lambda}_i -\gamma \ln z^1 +\gamma \ln \left( w^1_i/w^1 \right) + \epsilon_i^1
\end{align*}
where, we assume that $\epsilon_i^0$ and $\epsilon_i^1$ are identically and normally distributed disturbance terms.
Subtraction results in the main regression equation of the following:
\begin{align} 
\Delta \ln {a}_{i}
= -\gamma \Delta \ln z + \gamma \Delta \ln \left( w_i/w \right) +\epsilon_i
\label{main}
\end{align}
Here, the disturbance term $\epsilon_i=\epsilon_i^0-\epsilon_i^1$ is identically normally distributed, so that one can estimate $\gamma$ and $\Delta \ln z$ via a simple linear regression (\ref{main}).
That is, by regressing the growth of factor-wise cost shares i.e, $\Delta \ln {a}_{i}$ on the growth of relative prices i.e., $\Delta \ln \left( w_i/w \right)$, the slope gives the estimate of $\gamma$ while the intercept gives the estimate of $-\gamma \Delta \ln z$.
Also, note that ${\lambda}_i$ can be calibrated via (\ref{cs}) as long as we have the estimate for $\gamma$.

\subsection{The Data and Estimation}
A set of linked input--output tables includes sectoral transactions in both nominal and real terms. 
Since real value is adjusted for inflation, in order to enable comparison of quantities as if prices had not changed, and since nominal value is not adjusted, we use a price index to convert nominal into real values. 
That is, if we standardize the value of a commodity at the reference state as real, its nominal (unadjusted) value at the target state relative to the reference state equals the price index called a deflator.
Naturally, the 1995--2000--2005 linked input--output tables for both Japan \citep{miac} and Korea \citep{bok} include factor-wise deflators (395 factors for Japan and 350 factors for Korea) spanning the fiscal years recorded. 
These linked input--output tables, however, do not include deflators for primary factor (i.e., labor and capital) and therefore, we used the quality-adjusted price indexes compiled by \citet{jip} for Japan and by \citet{kip} for Korea in order to inflate the primary factor inputs observed in nominal values.

Hence, observations for both the dependent variables (cost shares as input--output coefficients $a_{ij}$) and the independent variables (price ratios $w_j/w_i$) for estimating (\ref{main}) become available with sufficient capacity, in terms of degrees of freedom, as we verify that there are $n+1$ inputs: namely, $i=0, 1, \cdots,n$; and $n$ outputs, namely $j=1,\cdots,n$, for an input--output table. 
In particular, we use the 2000 and 2005 input--output coefficient matrices out of the three-period linked input--output tables as the data for the cost share growth (i.e., $\Delta \ln a_{ij}$) and as we set the reference state at year 2000, the five-year growth of output-relative factor prices becomes simply the log difference between deflators; that is,  
\begin{align*}
\Delta \ln w_i/w_j = \ln p_i/p_j
\end{align*}
where $p_i$ denotes the deflator for commodity $i$ in year 2005 with respect to year 2000.

Figure \ref{spjpn} displays the estimated CES elasticity (i.e., $\sigma_j = 1 - \gamma_j$) with respect to the statistical significance of $\gamma_j$ i.e., the slope of the regression equation (\ref{main}) in terms of P-value, for Japan.
Figure \ref{spkor} is the version for Korea.
Note that CES elasticities were statistically significant (P-value $< 0.1$) for 176 out of 395 sectors for Japan, whereas 166 sectors were significant out of 350 sectors for Korea.
The results of estimation are summarized in the Appendix, Tables \ref{tab_JPN} and \ref{tab_KOR} for Japan and Korea, respectively.
These tables are confined to sectors whose slopes ($\gamma_j = 1-\sigma_j$) of the regression (\ref{main}) are statistically significant, and we indicate the level of significance by *** (0.01 level), **(0.05 level), and *(0.1 level), along with the estimated elasticities.
\begin{figure}[t!]
 \centering
  \includegraphics[scale = 0.63]{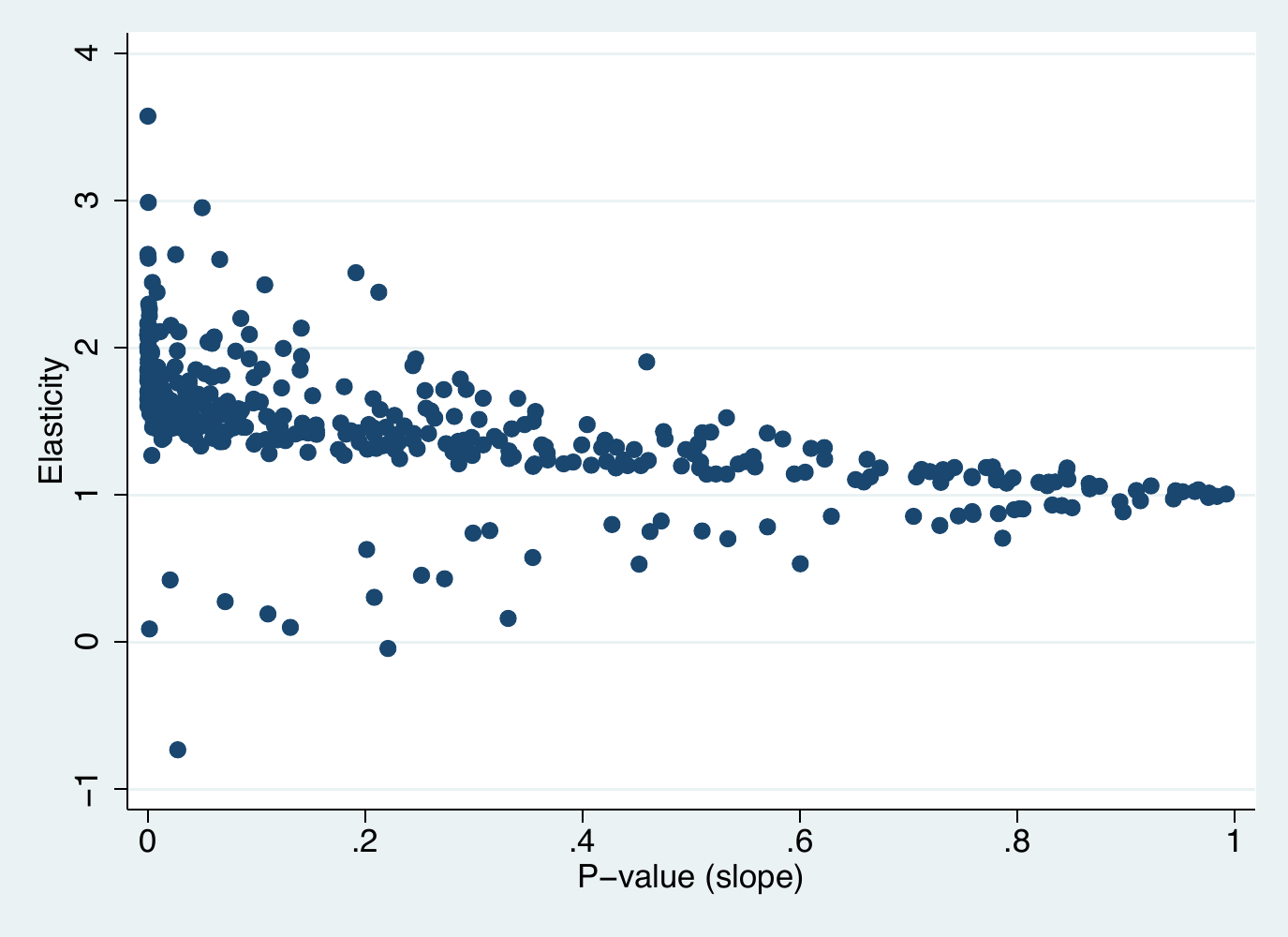}
 \caption{CES elasticity vs significance (Japan)} \label{spjpn}
\vspace{5mm}
  \includegraphics[scale = 0.63]{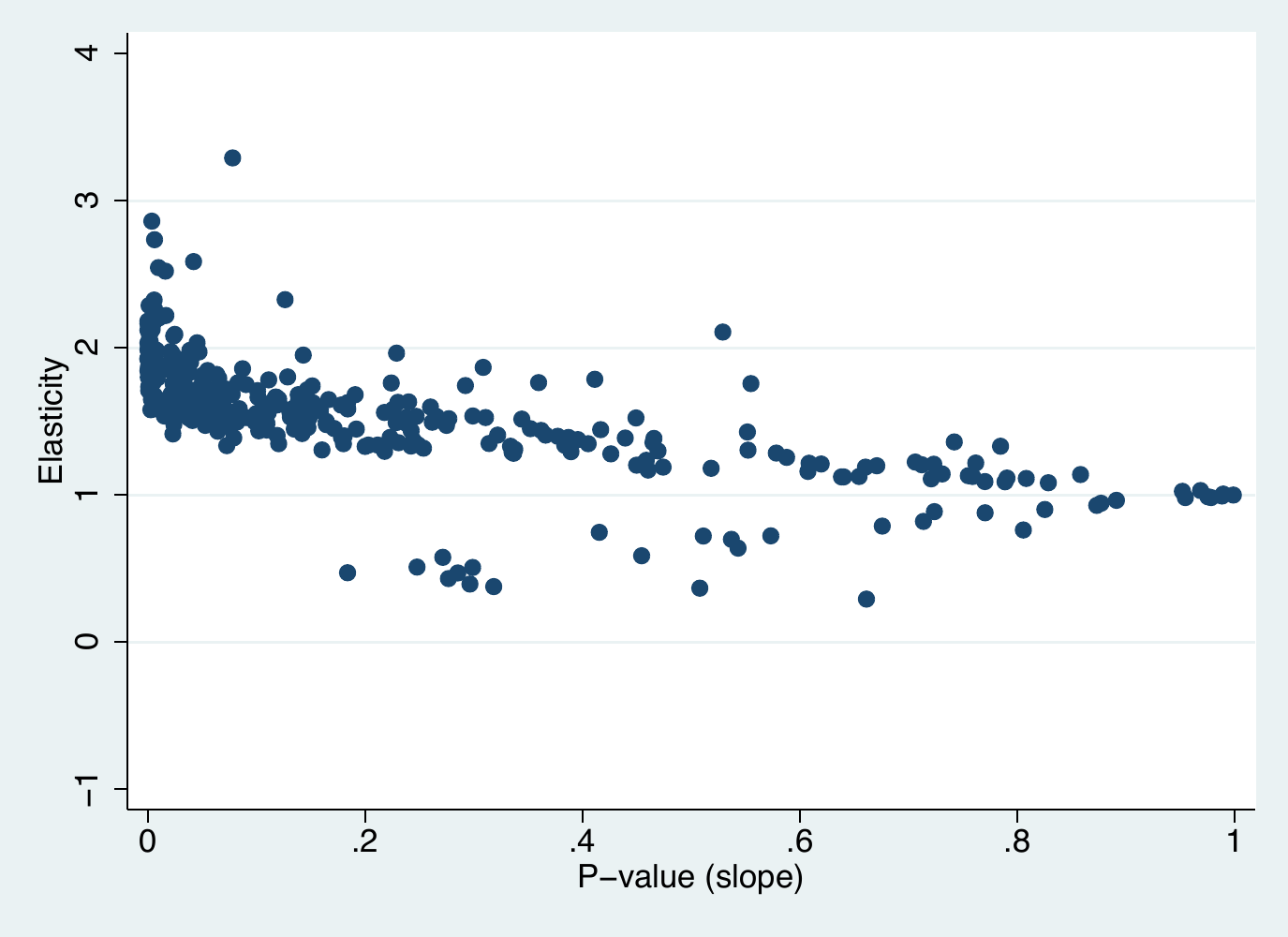}
 \caption{CES elasticity vs significance (Korea)} \label{spkor}
\vspace{-2mm}
\end{figure}
Note that we accept the null (i.e., $\gamma_j = 1-\sigma_j =0$) for sectors with statistically insignificant slope, and in that event, the average of elasticities i.e., $\sum_{j=1}^n \sigma_j/n$ was $1.32$ for Japan and $1.39$ for Korea.
Alternatively, if we accept all the estimates of elasticities, regardless of statistical significances, the average of elasticities was $1.46$ for Japan and $1.52$ for Korea.

These multifactor CES elasticities are comparable to other estimates in the literature.
The \citet{gtap} substitution elasticities for intermediate inputs which are broadly employed in CGE studies \citep[e.g., ][]{am, antimiani} range from 0.20 to 1.68, while those among internationally traded goods (i.e., Armington elasticities) are generally larger ranging from 1.15 to 34.40, depending on the industrial sector.
\citet{welsch}'s estimate for mean Armington elasticities ranges from negative 2.06 to positive 2.17, also depending on the industrial sector.
Note that these estimates are fairly comparable to the \citet{KS_2015}'s KLEM nest-wise CES elasticity estimates for 36 industrial sectors. 

In the third column of Tables \ref{tab_JPN} and \ref{tab_KOR}, we display the productivity growth $\Delta \ln z$, labeled as TFPg (Total Factor Productivity growth), which is the estimated constant of (\ref{main}) divided by the negative of the corresponding slope.  
Accordingly, the statistical significances of TFPg are evaluated by way of bootstrapping (with 400 replications) on the basis of regression (\ref{main}).   
The statistical significances of the underlying intercept are indicated with parenthesis.
Note also that these tables are sorted by the level of the estimated TFPg.
Let us now make some assessments of the estimated TFPg in regard to other possible productivity measurements.
Below is the log of T{\"o}rnqvist index 
\begin{align}
\text{TFPg (Translog)}= - \ln p + \sum_{i=0}^{n} \left(\frac{a_i^0 + a_i^1}{2}\right) \ln p_i
\label{tornq}
\end{align}
the exactness of which \citet{diewert} showed in measuring the productivity growth of Translog functions.
Thus, we know that (\ref{tornq}) is equal to the productivity growth of the underlying Translog function with or without knowing its parameters.  
Note that although it is almost impossible to estimate the parameters of a Translog function with one hundred factor inputs, its productivity growth can be measured using the same data (cost shares and price changes) as we use in estimating productivity for a multifactor CES function.
\citet{starhall} showed that the T{\"o}rnqvist index is a good approximation of TFPg measurement irrespective of the type of aggregator function and the interval of observations.
\begin{figure}[t!]
 \centering
  \includegraphics[scale = 0.63]{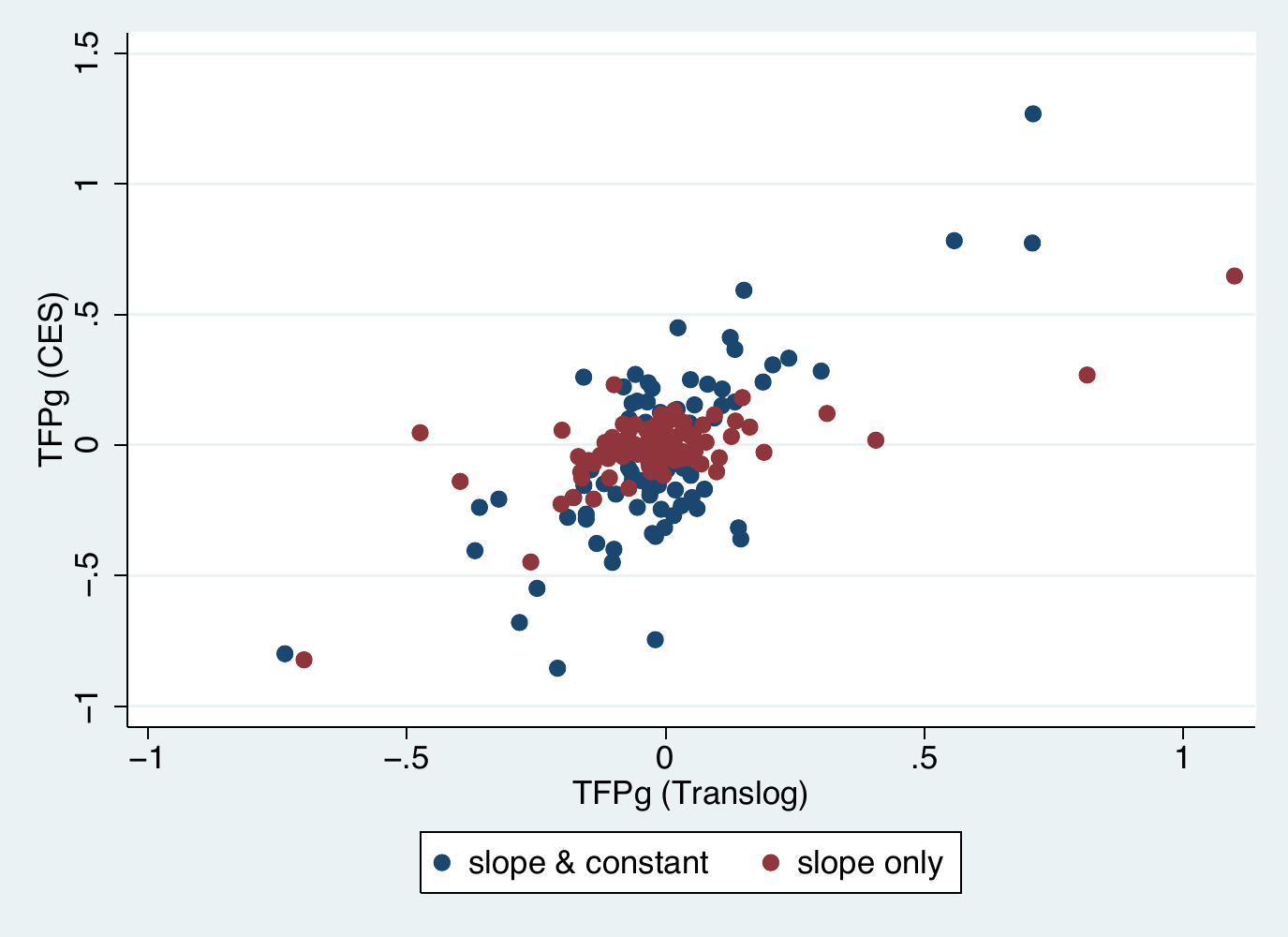}
 \caption{TFPg of different measurements. (Japan)} \label{tfp-tfp_JPN}
\vspace{5mm}
  \includegraphics[scale = 0.63]{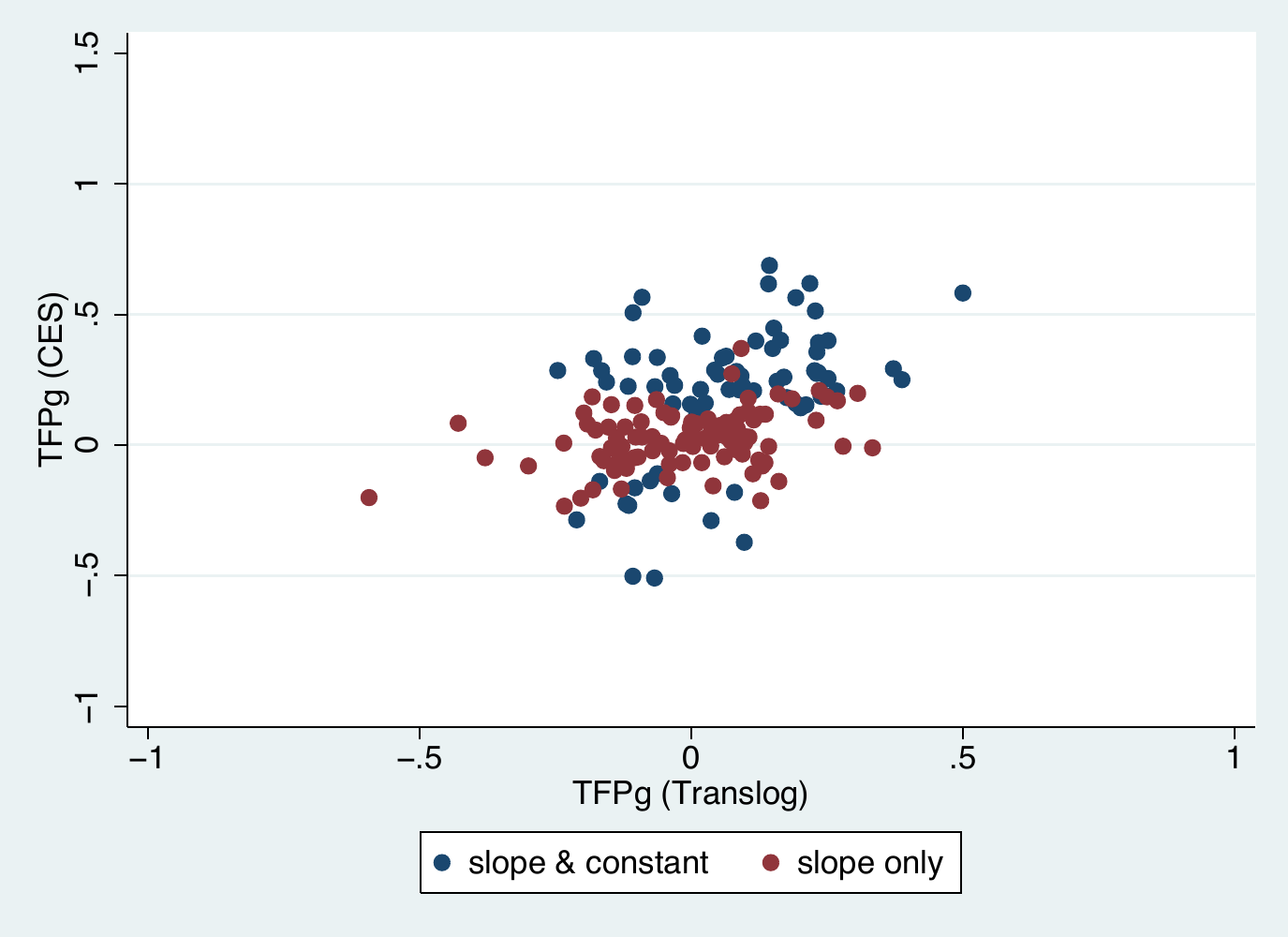}
 \caption{TFPg of different measurements. (Korea)} \label{tfp-tfp_KOR}
\vspace{-2mm}
 \end{figure}

In Figures \ref{tfp-tfp_JPN} and \ref{tfp-tfp_KOR}, we plot the estimated TFPg for a multifactor CES function, tagged as TFPg (CES) for all sectors listed in Tables \ref{tab_JPN} and \ref{tab_KOR}, against the log of the T{\"o}rnqvist indexes, tagged as TFPg (Translog).
Blue dots indicate sectors whose slope and intercept of regression (\ref{main}) were both statistically significant (P-value $<0.1$), whereas red dots indicate sectors with a slope that was significant but an intercept that was not.
In both cases, we observe agreements between the two TFPg measurements; therefore, we evaluate them objectively as summarized in Table \ref{tab_cc}.
Here, Correlation designates Pearson's correlation coefficient, whereas Concordance designates  \citet{lin}'s concordance correlation coefficient.
Note that ``Slope Only'' significant sectors are of red dots, and ``Sope and Constant'' significant sectors are of blue dots, in Figures \ref{tfp-tfp_JPN} and \ref{tfp-tfp_KOR}; and ``Slope'' indicates all slope significant sectors, thus, union of red and blue sectors.
``Bootstrapping'' indicates sectors with significant TFPg (CES) estimates via bootstrapping.
To say this in other words, by way of multifactor CES function, we obtain TFPg estimates similar to those based on Translog functions that are very general in terms of elasticities of substitution set aside their estimability, and yet, a multifactoral elasticity of substitution is estimable over very many factor inputs.
Note however that in the event that we accept the null for the insignificant slope of the regression (\ref{main}), we must assume that the function is Cobb--Douglas and that TFPg is unmeasureable.
\newcolumntype{d}{D{.}{.}{4}}
\begin{table}[htb!]
\center
\caption{Concordances and correlations between Translog and multifactor CES TFPg estimates.}
\label{tab_cc}
\begin{tabular}{lddr}
\hline
\multicolumn{1}{c}{Sectors} &	\multicolumn{1}{c}{Concor.}  	&	\multicolumn{1}{c}{Correl.}  &\multicolumn{1}{c}{Obs.} \\ \hline
Slope (JPN) & 0.645 & 0.669 & 176 \\
Slope Only (JPN) 	& 	0.673		&	0.707 & 100 \\
Slope and Constant (JPN) & 0.633 & 0.741 & 76 \\ 
Bootstrapping (JPN) & 0.794 & 0.889 & 21 \\
Slope (KOR) & 0.305 & 0.413 & 166 \\
Slope Only (KOR) 	& 	0.309		&	0.340 & 97 \\
Slope and Constant (KOR) & 0.370 & 0.413 & 69 \\
Bootstrapping (KOR) & 0.623 & 0.707 & 33 \\
\hline
\end{tabular}
\end{table}

\section{Prospective Analysis}
\subsection{Projected Prices}
In the following section, we construct a multisectoral general equilibrium model that reflects all measured elasticities and observed current cost shares; further, we exogenously impose some productivity change into the model and simulate the multisectoral propagation that can potentially take place.
For sake of simplicity, let us normalize all current prices at unity.
In that event, we know by (\ref{cs}) that: 
\begin{align*}
a_{ij} = {\lambda}_{ij}, ~~~~~~\sum_{i=0}^n{a_{ij}} =1, ~~~~j=1,2,\cdots,n
\end{align*}
Then, the system of CES unit cost functions in equilibrium, under some exogenously given productivity change i.e., $\mathbf{z}=\left( z_1, z_2, \cdots, z_n \right) \neq \boldsymbol{1}$, must be in the following state:
\begin{align}
\begin{split}
\pi_1 &= z_1^{-1}\left( a_{01} \pi_0^{{\gamma}_1} + a_{11} \pi_1^{{\gamma}_1} + \cdots a_{n1} \pi_n^{{\gamma}_1} \right)^{\frac{1}{{\gamma}_1}} \\
\pi_2 &= z_2^{-1}\left( a_{02} \pi_0^{{\gamma}_2} + a_{12} \pi_1^{{\gamma}_2} + \cdots a_{n2} \pi_n^{{\gamma}_2} \right)^{\frac{1}{{\gamma}_2}} \\
&\vdots \\
\pi_n &= z_n^{-1}\left( a_{0n} \pi_0^{{\gamma}_n} + a_{1n} \pi_1^{{\gamma}_n} + \cdots a_{nn} \pi_n^{{\gamma}_n} \right)^{\frac{1}{{\gamma}_n}} 
\end{split} \label{ge}
\end{align}
where the projected (ex post) general equilibrium price for factor $i$ is denoted by $\pi_i$. 
Note that the current state i.e, $\mathbf{z}=\boldsymbol{1}$ can be reproduced by setting all prices at the current state i.e., $\boldsymbol{\pi} =\boldsymbol{1}$ and vice versa.\footnote{This may not be so obvious when $\boldsymbol{\gamma} = \mathbf{0}$, until we see (\ref{lnpi}).} 

The projected price, ex post the exogenous productivity change, can be obtained by solving (\ref{ge}) for $\boldsymbol{\pi}$. 
By rearranging, we have:
\begin{align*}
z_1^{{\gamma}_1} \pi_1^{{\gamma}_1} &=  a_{01} \pi_0^{{\gamma}_1} + a_{11} \pi_1^{{\gamma}_1} + \cdots a_{n1} \pi_n^{{\gamma}_1} \\
z_2^{{\gamma}_2}\pi_2^{{\gamma}_2} &= a_{02} \pi_0^{{\gamma}_2} + a_{12} \pi_1^{{\gamma}_2} + \cdots a_{n2} \pi_n^{{\gamma}_2} \\
&\vdots \\
z_n^{{\gamma}_n}\pi_n^{{\gamma}_n} &=  a_{0n} \pi_0^{{\gamma}_n} + a_{1n} \pi_1^{{\gamma}_n} + \cdots a_{nn} \pi_n^{{\gamma}_n}  
\end{align*}
or by way of row vectors and matrices:
\begin{align*}
\boldsymbol{\pi}^{\boldsymbol{\gamma}} \left< \mathbf{z}^{\boldsymbol{\gamma}} \right>
=\mathbf{a}_{0} + \boldsymbol{\pi}^{\boldsymbol{\gamma}} \mathbf{A}
\end{align*}
where $\boldsymbol{\pi}^{\boldsymbol{\gamma}}=\left( \pi_1^{\gamma_1}, \cdots, \pi_n^{\gamma_n} \right)$ and $\mathbf{z}^{\boldsymbol{\gamma}}=\left( z_1^{\gamma_1}, \cdots, z_n^{\gamma_n} \right)$, while we set the price of a primary input as a num\'{e}raire i.e., $\pi_0=1$.
Angle brackets indicate diagonalization.
Note that $\mathbf{A}$ and $\mathbf{a}_0$ are the current input--output coefficients matrix and value added coefficients vector, respectively.
Now, the projected equilibrium price $\boldsymbol{\pi}$ can be obtained in terms of $\mathbf{z}$:
\begin{align}
\boldsymbol{\pi} = \left( \mathbf{a}_0 \left[ \left< \mathbf{z}^{\boldsymbol{\gamma}} \right> - \mathbf{A}  \right]^{-1} \right)^{\frac{1}{\boldsymbol{\gamma}}}
\label{pi_ces}
\end{align}

Besides CES, we may use (\ref{pi_ces}) to obtain the projected price for the cases of Leontief ($\boldsymbol{\gamma} = \mathbf{1}$) and Cobb--Douglas ($\boldsymbol{\gamma} = \mathbf{0}$).
The Leontief case is straightforward:
\begin{align}
\boldsymbol{\pi} = \mathbf{a}_0 \left[ \left< \mathbf{z} \right> - \mathbf{A}  \right]^{-1} 
\label{pi_leontief}
\end{align}
For the Cobb--Douglas case, we first take the log of (\ref{ge}) and then let $\boldsymbol{\gamma}\to \mathbf{0}$.
Below, we work on the unit cost function of any industrial sector $j$ while omitting the subscript:
\begin{align*}
\ln \pi + \ln z 
=\frac{\ln \left( a_0 + \sum_{i=1}^n a_{i} \pi_i^{\gamma} \right)}{\gamma} 
\to \sum_{i=1}^n a_i \ln \pi_i
\end{align*}
Here, we applied l'Hospital's rule when we let $\gamma \to 0$, since in that event the nominator and the denominator both approach zero.
By way of row vectors and matrices, this can be written concisely:
\begin{align}
\ln \boldsymbol{\pi} = -\ln \mathbf{z} + \left(\ln \boldsymbol{\pi} \right) \mathbf{A}
\label{lnpi}
\end{align}
where the log operators are applied element-wise.
The Cobb--Douglas version of the projected price will thus be:
\begin{align}
\boldsymbol{\pi} &= \exp \left( - \left(\ln \mathbf{z} \right) \left[ \mathbf{I} - \mathbf{A} \right]^{-1} \right) 
\label{pi_cd}
\\
&= \left( \frac{1}{\prod_{i=1}^{n} z_i^{\ell_{i1}}}, \frac{1}{\prod_{i=1}^{n} z_i^{\ell_{i2}}}, \cdots, \frac{1}{\prod_{i=1}^{n} z_i^{\ell_{in}} }\right) \notag
\end{align}
where, $\ell_{ij}$ is an element of the Leontief inverse matrix $\left[ \mathbf{I} - \mathbf{A} \right]^{-1}$.

\subsection{Projected Structures}
Since we set the current price to unity, the final demand in monetary terms will be the same as the physical quantity demanded.
Let the current (nominal) final demand be denoted by a column vector $\mathbf{d} = \left( d_1, \cdots, d_n \right)^\intercal \geq \mathbf{0}$.
Note that the sum of product-wise final demand and that of sector-wise value added (social cost) equals the GDP.
If we have the projected price attributable to some exogenous productivity change, we can evaluate the corresponding welfare change in terms of social cost saved (SCS, hereafter); that is,
\begin{align}
\text{SCS}
= \left( \mathbf{1} - \boldsymbol{\pi} \right) \mathbf{d}
=\sum_{j=1}^n v_j - v^{\prime}_j 
\label{scs}
\end{align}
Note that $v_j$ and $v_j^{\prime}$ denote current and projected value added for sector $j$.
The sector-wise distribution of SCS, however, requires more examination of the projected structure of the economy.

According to (\ref{cs}), the projected cost shares ex post the exogenous productivity change $\mathbf{z}$, which we denote by $b_{ij}$, can be evaluated by the following identity:
\begin{align}
b_{ij} = a_{ij} \left( z_j\pi_j/\pi_{i}\right)^{-\gamma_j}
~~~~~~~i=0,1,\cdots,n
\label{bij}
\end{align}
Hence, under CES, the projected primary factor input (or value added) distribution $\mathbf{v}^{\prime}=\left(v^{\prime}_1, \cdots,v^{\prime}_n\right)$ spanning over the sectors for a given fixed final demand $\mathbf{d}$ (in physical quantity) can be evaluated as follows:
\begin{align}
\mathbf{v}^{\prime} 
&= \mathbf{b}_{0} \left[ \mathbf{I} - \mathbf{B} \right]^{-1} \left< \boldsymbol{\pi} \right> \left< \mathbf{d} \right>
\label{vprime_ces}
\end{align}
where the entries for $\mathbf{b}_0$ and $\mathbf{B}$ are specified by (\ref{bij}).
Conversely, the current distribution of primary factor inputs (or value added) $\mathbf{v}=\left(v_1, \cdots,v_n\right)$ is specified by the current observed cost shares as follows:
\begin{align}
\mathbf{v} = \mathbf{a}_{0} \left[ \mathbf{I} - \mathbf{A} \right]^{-1} \left< \mathbf{d} \right>
\label{v_ces}
\end{align}
Since (\ref{vprime_ces}) and (\ref{v_ces}) are row vectors, one can evaluate SCS in terms of sector-wise distribution.

\subsection{Uniform CES Elasticity}
Here, we examine how SCS will be distributed among sectors depending on the projected structures pertaining to uniform substitution elasticities i.e., $\gamma_1=\gamma_2=\cdots=\gamma_n=\gamma$.
First, by plugging (\ref{bij}) into (\ref{vprime_ces}) under some uniform elasticity $\sigma=1-\gamma$,
we have the following exposition for the projected value added distribution:
\begin{align}
\mathbf{v}^{\prime} 
&= \mathbf{a}_{0} \left< \boldsymbol{\pi}^{-\gamma} \right> \left< \mathbf{z}^{-\gamma} \right>  \left[ \mathbf{I} - \left< \boldsymbol{\pi}^{\gamma} \right> \mathbf{A}\left< \boldsymbol{\pi}^{-\gamma} \right>\left< \mathbf{z}^{-\gamma} \right> \right]^{-1}  \left< \boldsymbol{\pi} \right> \left<\mathbf{d} \right> \notag \\
&= \mathbf{a}_{0}   \left[ \left< \mathbf{z}^{\gamma} \right> - \mathbf{A} \right]^{-1}  \left< \boldsymbol{\pi}^{1-\gamma} \right>
\left<\mathbf{d} \right>
\label{vprime_ces2}
\end{align}
Hence, we know that for Cobb--Douglas and Leontief cases the projected value added distribution will be:
\begin{align}
\mathbf{v}^{\prime}({\text{Cobb--Douglas}})
&= \mathbf{a}_{0}   \left[ \mathbf{I} - \mathbf{A} \right]^{-1}  \left< \boldsymbol{\pi} \right>\left<\mathbf{d} \right> \label{vcd}
\\
\mathbf{v}^{\prime}({\text{Leontief}})
&= \mathbf{a}_{0}   \left[ \left< \mathbf{z} \right> - \mathbf{A} \right]^{-1}  \left<\mathbf{d} \right>
\label{vl}
\end{align}
Note that projected equilibrium price (\ref{pi_cd}) must be applied to (\ref{vcd}) for the Cobb--Douglas case.

Further, let us show below that, under uniform substitution elasticity less than unity, the SCS distribution will always be positive (in all sectors) against any exogenous productivity increase, and vice versa.
Specifically, we show that
\begin{propo*}\normalfont
Under $0\leq \gamma \leq 1$, SCS is positive in all sectors such that $\mathbf{v}-\mathbf{v}' \geq \mathbf{0}$, if the exogenous productivity is increasing i.e., $\mathbf{z}\geq\mathbf{1}$, and SCS is negative in all sectors such that $\mathbf{v}-\mathbf{v}' \leq \mathbf{0}$, if the exogenous productivity is decreasing i.e., $\mathbf{z} \leq\mathbf{1}$.
\end{propo*}
\begin{proof}
Because the input--output coefficient as well as the productivity is nonnegative i.e., $\mathbf{A}\geq\mathbf{0}$ and $\mathbf{z}\geq\mathbf{0}$, we have the following exposition:
\begin{align}
\begin{split}
\left[ \left< \mathbf{z}^\gamma \right> - \mathbf{A}  \right]^{-1}
&= \left<\mathbf{z}^{-\gamma}\right> + \mathbf{A}\left<\mathbf{z}^{-2\gamma}\right> + \mathbf{A}^2\left<\mathbf{z}^{-3\gamma}\right> +\cdots 
\\ 
\left[ \mathbf{I} - \mathbf{A}  \right]^{-1}
&= \mathbf{I} + \mathbf{A} + \mathbf{A}^2 +\cdots 
\end{split} \label{spread}
\end{align}
Thus, by taking $0\leq \gamma \leq 1$ into account, we know that
\begin{align}
\begin{split}
\left[ \left< \mathbf{z}^\gamma \right> - \mathbf{A}  \right]^{-1}
&\leq \left[ \mathbf{I} - \mathbf{A}  \right]^{-1} ~~~~~\text{ if }~\mathbf{z} \geq \mathbf{1} \\
\left[ \left< \mathbf{z}^\gamma \right> - \mathbf{A}  \right]^{-1}
&\geq \left[ \mathbf{I} - \mathbf{A}  \right]^{-1} ~~~~~\text{ if }~\mathbf{z}\leq\mathbf{1}
\end{split}
\label{dst}
\end{align}
Moreover, as we take for granted that the unit cost mapping (\ref{ge}) is monotone increasing in price, the projected equilibrium price $\boldsymbol{\pi}$ must be smaller (larger) than unity when the exogenous productivity $\mathbf{z}$ is increasing (decreasing).
Thus, by taking $0\leq \gamma \leq 1$ into account we know that
\begin{align}
\begin{split}
& \boldsymbol{\pi} \leq \boldsymbol{\pi}^{1-\gamma} \leq \mathbf{1}  ~~~~~\text{ if }~\mathbf{z} \geq \mathbf{1} \\
& \boldsymbol{\pi} \geq \boldsymbol{\pi}^{1-\gamma} \geq \mathbf{1}  ~~~~~\text{ if }~\mathbf{z} \leq \mathbf{1}
\end{split}\label{dpr}
\end{align}
Hence, the structural differences between the reference and the projected states can be assessed as follows:
\begin{align}
\begin{split}
&\left[ \mathbf{I} - \mathbf{A} \right]^{-1} 
\geq
\left[ \left< \mathbf{z}^{\gamma} \right> - \mathbf{A} \right]^{-1}  \left< \boldsymbol{\pi}^{1-\gamma} \right> 
 ~~~~~\text{ if }~\mathbf{z} \geq \mathbf{1} 
 \\
&\left[ \mathbf{I} - \mathbf{A} \right]^{-1} 
\leq
\left[ \left< \mathbf{z}^{\gamma} \right> 
-
\mathbf{A} \right]^{-1}  \left< \boldsymbol{\pi}^{1-\gamma} \right> 
 ~~~~~\text{ if }~\mathbf{z} \leq \mathbf{1} 
\end{split} \label{dmain}
\end{align}
Since the SCS distribution $\mathbf{v}-\mathbf{v}^\prime$ is the difference between (\ref{v_ces}) and (\ref{vprime_ces2}), the above (\ref{dmain}) suffices for the proposition.
\end{proof}
\begin{rem*}\normalfont
This proposition is inclusive of Cobb--Douglas ($\gamma=0$) and Leontief ($\gamma=1$) systems.
Uniformity of substitution elasticity $\gamma$ is required for obtaining (\ref{vprime_ces2}).
For substitution elasticity larger than unity i.e., $\gamma = 1-\sigma < 0$, the inequalities for (\ref{dst}) will be reversed whereas those for (\ref{dpr}) remain stable, so that (\ref{dmain}) may not hold necessarily.
\end{rem*}

\subsection{Simulation}
Let us now apply the framework specified in the previous sections.
First, we calibrate the multisectoral models with different elasticities, namely Leontief, Cobb--Douglas, and multifactor CES, as of year 2005.
Thus, the cost shares of the current state i.e. $\mathbf{a}_0$ and $\mathbf{A}$ are as of year 2005. 
For the multifactor CES system, we make use of the elasticities that were statistically significant i.e., the sectors displayed in Tables \ref{tab_JPN} and \ref{tab_KOR}, while we undertake unit elasticity (or the null hypothesis) for the rest of the sectors.\footnote{For sake of reference, we may also use the estimated elasticities for all sectors, regardless of statistical significances.
Such case will be indicated as CES (all estimates), henceforth.}

As for the exogenous productivity change $\mathbf{z}$, we examine the ``productivity doubling'' of the ``Ready mixed concrete'' (RMC, hereafter) sector which is 150th sector for Japan, and the 159th for Korea.
That is, 
\begin{align}
\begin{split}
\text{Japan:}~~~&z_{j=150} = 2,~~z_{j\neq150}=1~~~~(n=395) \\
\text{Korea:}~~~&z_{j=159} = 2,~~z_{j\neq159}=1~~~~(n=350)
\end{split} \label{rmc}
\end{align}
There are couple reasons for choosing this sector.
For one thing, this stimuli is better influential than not throughout the economy.
In other words, upstream industrial sectors are preferable, for they may be influential to all downstream sectors, whereas downstream sectors do not have much influence on upstream sectors. 
We performed triangulation,\footnote{Stages of production leading to final goods are investigated through permutation of sectors. See, \citet{kondo} for recent developments.} 
in regard to the work of \citet{cw}, upon the 2005 input--output coefficient matrices for both Japan and Korea, and we found that the RMC sector was placed at the upper stream (137th out of 395 for Japan, and 65th out of 350 for Korea) of the supply chain in both economies.
Another criterion is whether the output of the sector is completely domestic (non imported) as the current study precludes international trade.
And most importantly, the equivalence of the sector to be examined for the two countries is required.
The RMC sector meets all of these criteria.

In Table \ref{tab_scs} we summarize the results of calculating SCS via (\ref{scs}) for the four systems: namely Leontief, Cobb--Douglas, CES, and CES (all estimates); in two countries: namely Japan and Korea. 
\begin{table}[tb!]
\center
\caption{SCS (social cost saved) by productivity doubling of RMC (ready mixed concrete) sector.
BJPY stands for Billion Japanese Yens.
BKRW stands for Billion Korean Republic Wons.
Values in parentheses are the kurtosis of the corresponding SCS distribution.}
\label{tab_scs}
\begin{tabular}{lrcrc}
\hline
\multicolumn{1}{c}{} 
&	\multicolumn{2}{c}{Japan [BJPY]} 
&\multicolumn{2}{c}{Korea [BKRW]} 
\\ \hline
Output	&	1,347	&		&	6,398	&		\\
SCS Leontief	&	674	&	(315)	&	3,203	&	(162)	\\
SCS Cobb--Douglas	&	926	&	(52)	&	4,349	&	(84)	\\
SCS CES	&	944	&	(45)	&	4,550	&	(102)	\\
SCS CES (all estimates)   &   976  & (39) & 4,643 & (75) \\
\hline
\end{tabular}
\end{table}
The projected equilibrium price $\boldsymbol{\pi}$ for given $\mathbf{z}$ as in (\ref{rmc}) is calculated using (\ref{pi_leontief}) for the Leontief, (\ref{pi_cd}) for the Cobb--Douglas, and (\ref{pi_ces}) for the CES systems.
Along with the SCS, we display the output of the RMC sector of the 2005 input--output table. 
Notably, the SCS of the Leontief system is very slightly larger than one half the output of the RMC sector, reflecting the productivity doubling of the RMC sector.
This is legitimate, in regard to (\ref{spread}), as we consider:
\begin{align*}
\left[ \mathbf{I} - \mathbf{A} \right]^{-1} 
-\left[ \left< \mathbf{z} \right> - \mathbf{A} \right]^{-1}
\approx \mathbf{I} - \left< \mathbf{z} \right>^{-1} = 1/2
\end{align*}
Conversely, the SCS of the Cobb--Douglas and CES systems is larger than that of the Leontief system, reflecting further propagation across sectors that have larger elasticity.

Let us now look into the sectoral distribution of the SCS.
Figures \ref{scsdLJ}, \ref{scsdCDJ}, \ref{scsdCESJ}, and \ref{scsdCESrJ} show the projected sector-wise SCS from productivity doubling in the RMC sector under the Leontief, Cobb--Douglas, CES, and CES (all estimates) systems, respectively, for Japan. 
Corresponding figures for Korea are Figures \ref{scsdLK}, \ref{scsdCDK}, \ref{scsdCESK}, and \ref{scsdCESrK}.
As we have anticipated in regard to the previous Proposition, SCS for the Leontief and Cobb--Douglas systems is distributed on the positive side overall.\footnote{However, due to the negative entries for $\mathbf{d}$, slightly negative values are observed.}
At base, when there is productivity doubling in one sector, its price will be cut in half.
The inter-sectoral propagation of that price change will nevertheless be different, depending on the elasticity of factor substitution among the interacting sectors.
As for the Leontief system, because factor substitution will not exist in any other sector, the price change of RMC to half its former level will have no effect upon its intermediate demand.
Thus, in that event, all the factor inputs (including the primary factor) for the RMC sector will be reduced by half.
This is the main reason why the primary factor for the RMC sector is reduced (as SCS) rather prominently for the Leontief system.
Consequently, the intermediate demand of the factors (including the primary factor) will be reduced respectively by as much as half the amount that used to go into the RMC sector.
Such reduction of intermediate demand and thus of supply will be accumulated in convergence.
In other words, at least half of the primary factor put into the RMC sector will be directly reduced, and beyond that, the primary factor in any other sector will be reduced indirectly.
Figures \ref{scsdLJ} and \ref{scsdLK} reflect such propagation of productivity doubling in the RMC sector upon primary factor demand under a system of zero elasticity of substitution.

In contrast, as for the Cobb--Douglas system, the intermediate demand for RMC, when its price is reduced to half, must be doubled; that is the very definition of unit elasticity of substitution.
Thus, in that event, the monetary output and the factor inputs (including the primary factor) of the RMC sector will not change. 
As for an elastic CES system with elasticity of substitution larger than unity, the factor demand for RMC becomes larger than two fold, when the price of RMC is reduced by half.
And in that event, the factor inputs of the RMC sector can be increased.\footnote{This is the main reason why we observe, in Figures \ref{scsdCESJ} and \ref{scsdCESK}, negative SCS (increased primary factor input) in the RMC sector.}
In either system, since the system of unit cost functions is strictly concave, the price of all factors except that of the primary factor that will stay constant, will converge in a strictly descending manner.
Hence, in equilibrium, the primary factor will be mitigated for the sectors where the primary factor becomes relatively expensive compared with other factor inputs.
Notably, Figures \ref{scsdCDJ} and \ref{scsdCESJ} indicate that primary factor is reduced (as SCS) rather prominently at sectors, namely, ``Public construction of roads'' (279th), ``Public construction of rivers, drainages and others'' (280th), and ``Residential construction (non-wooden)'' (275th), for Japan.
Figures \ref{scsdCDK} and \ref{scsdCESK} indicate that ``Residential building construction'' (289th), ``Road construction'' (272nd), and ``Non-residential building construction'' (270th) are prominent for Korea.
These sectors are obviously the ones that utilize RMC extensively for production.
In other words, the primary factor in these sectors will be substituted by RMC with reduced price.

\begin{figure}[t!]
 \centering
  \includegraphics[scale = 0.63]{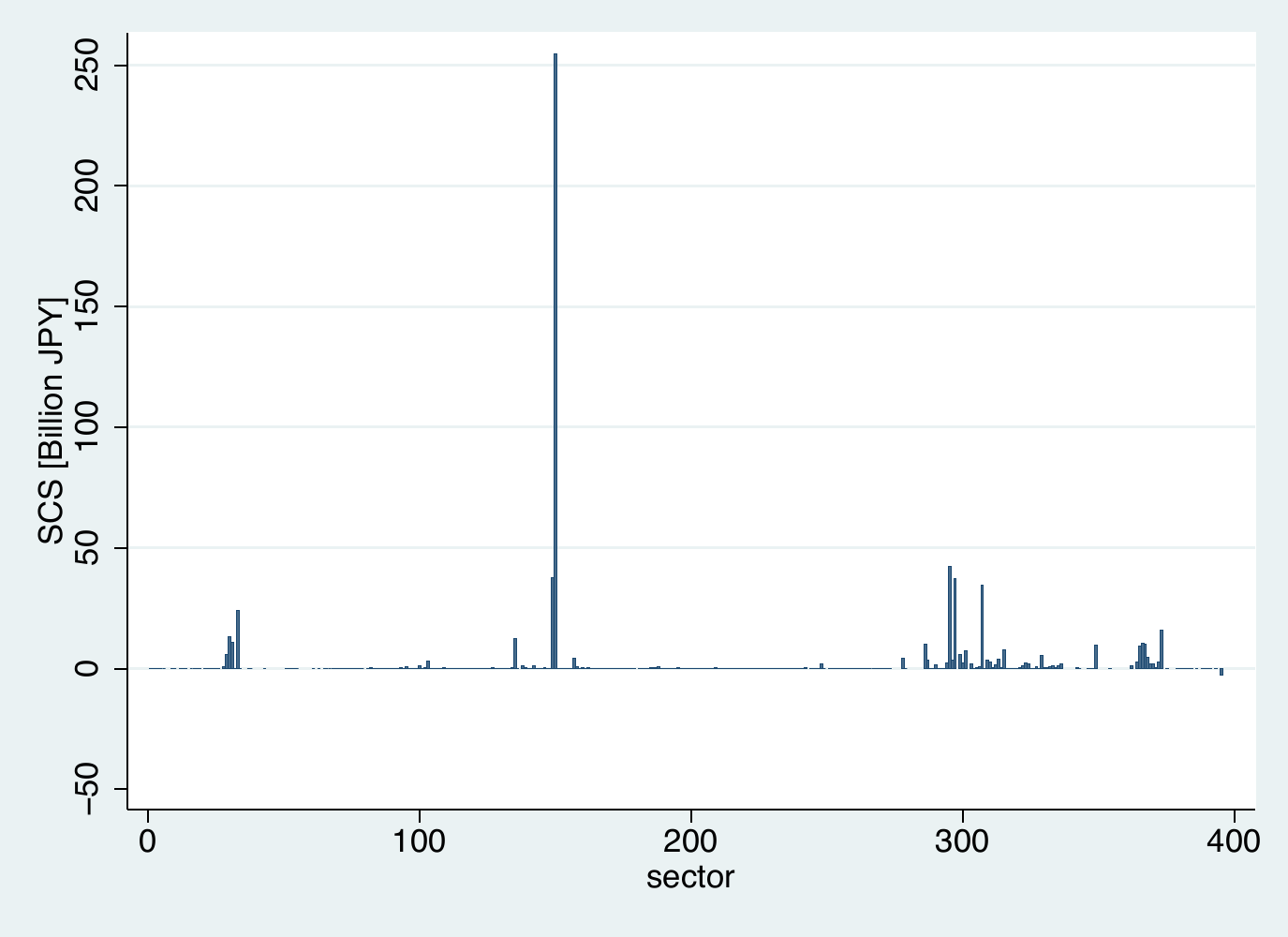}
 \caption{Sectoral distribution of SCS for productivity doubling of RMC sector (150th) for Leontief system. (Japan)} \label{scsdLJ}
\end{figure}

\begin{figure}[t!]
\centering
  \includegraphics[scale = 0.63]{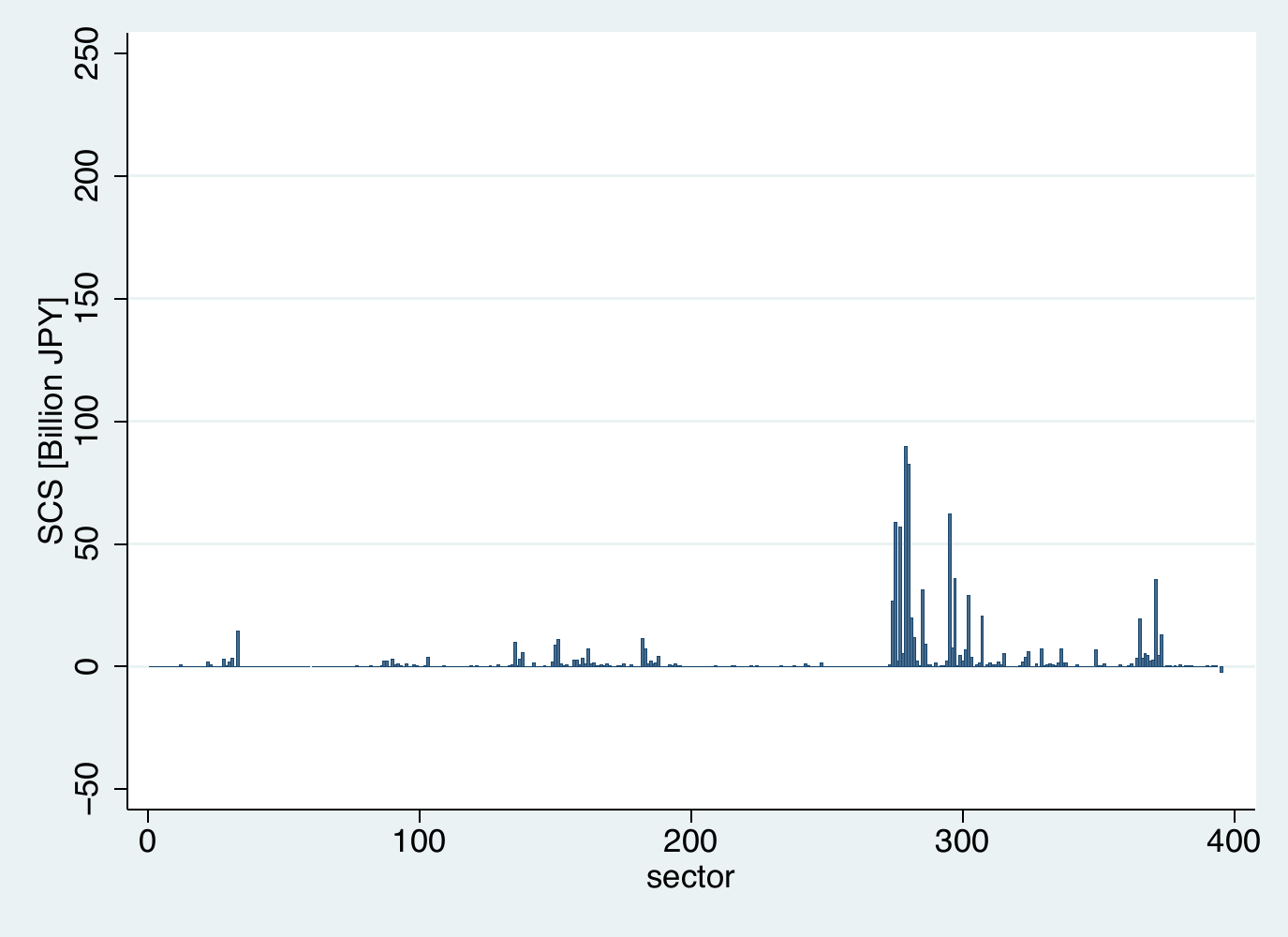}
 \caption{Sectoral distribution of SCS for productivity doubling of RMC sector (150th) for Cobb--Douglas system. (Japan)} \label{scsdCDJ}
\end{figure}

\begin{figure}[h!]
\centering
  \includegraphics[scale = 0.63]{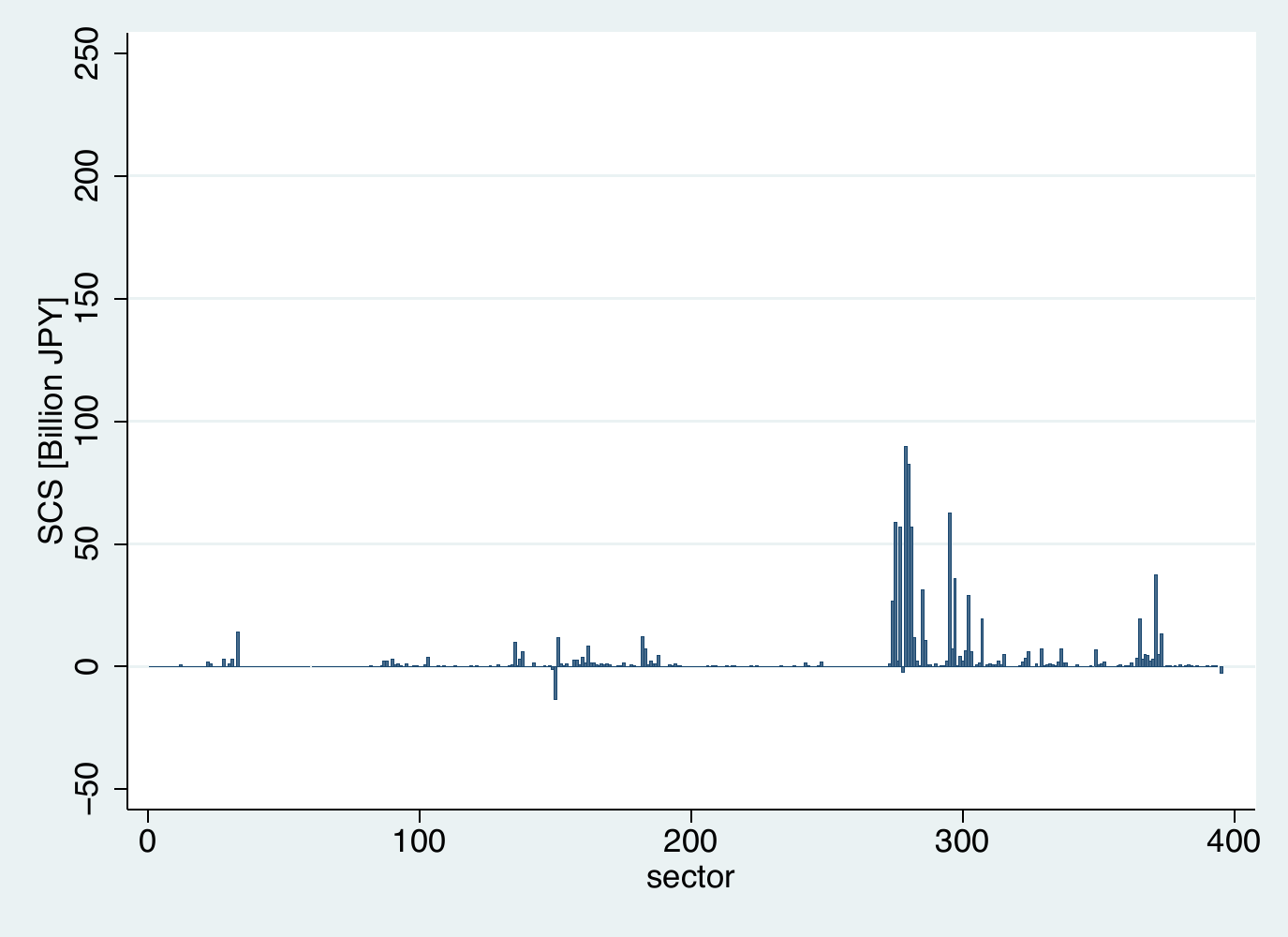}
 \caption{Sectoral distribution of SCS for productivity doubling of RMC sector (150th) for multifactor CES system. (Japan)} \label{scsdCESJ}
 \end{figure}

\begin{figure}[t!]
\centering
  \includegraphics[scale = 0.63]{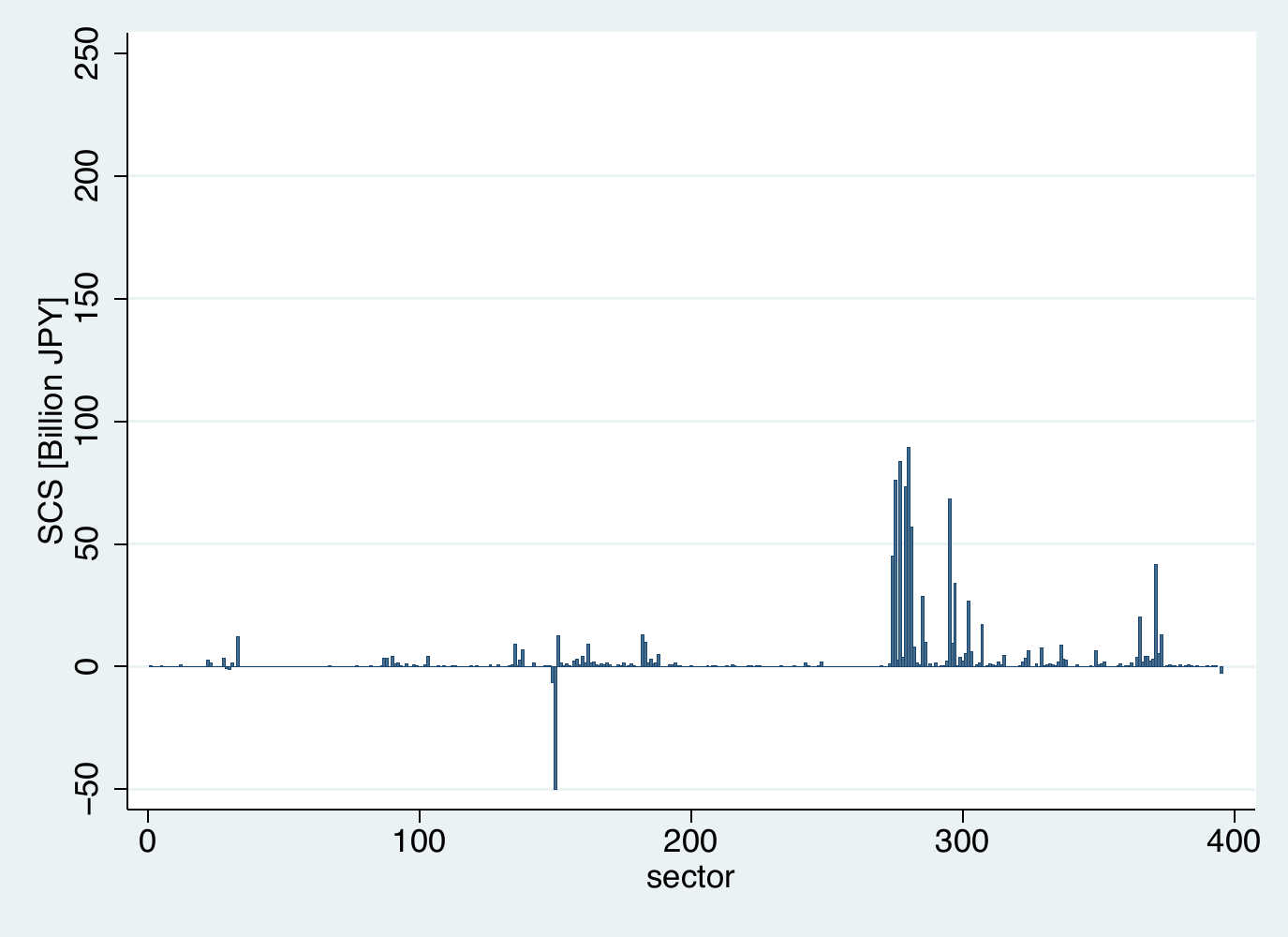}
 \caption{Sectoral distribution of SCS for productivity doubling of RMC sector (150th) for multifactor CES (all estimates) system. (Japan)} \label{scsdCESrJ}
 \end{figure}

Moreover, we observe from these figures that not only the magnitude of propagation (in terms of SCS) of the productivity stimuli will be magnified by larger elasticities of substitution, but the distribution of SCS become more even.
We have measured the ``polarity'' of the distribution of SCS over the sectors via kurtosis, displayed in parentheses in Table \ref{tab_scs}.
The primary factor will be mitigated primarily at the RMC sector where the productivity is enhanced for the Leontief system, whereas the mitigation of primary factor will spread over the sectors for the Cobb--Douglas and CES.
Put differently, the welfare gain of enhanced productivity in one industry is attained mainly as the curtailment of factor inputs of that particular industry while keeping the output level consistent, for the Leontief system, whereas for the Cobb--Douglas and CES systems the reduced price is appreciated by other industries so that their primary factors are reduced by substitution.

\begin{figure}[t!]
 \centering
  \includegraphics[scale = 0.63]{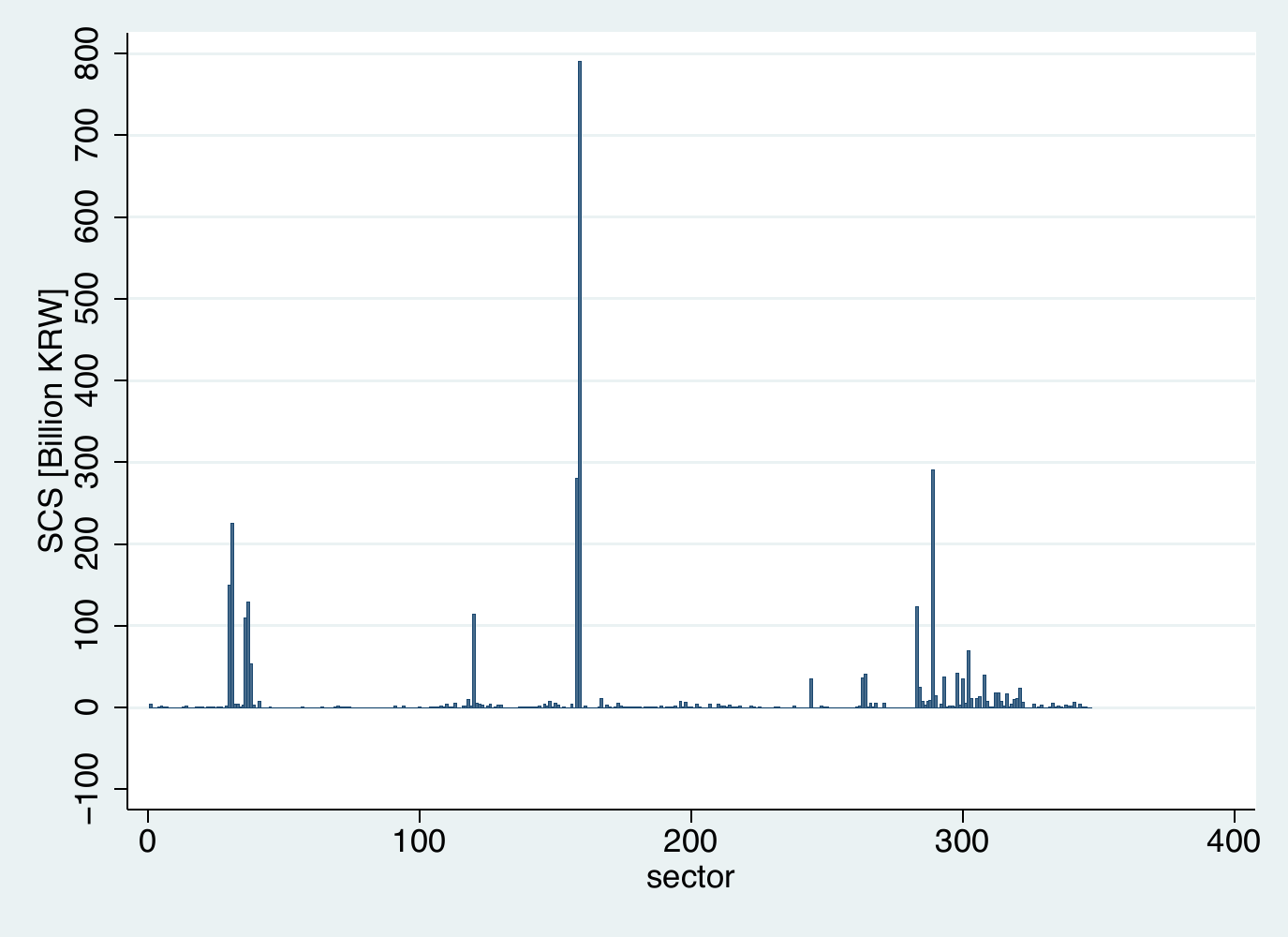}
 \caption{Sectoral distribution of SCS for productivity doubling of RMC sector (159th) for Leontief system. (Korea)} \label{scsdLK}
\end{figure}

\begin{figure}[t!]
\centering
  \includegraphics[scale = 0.63]{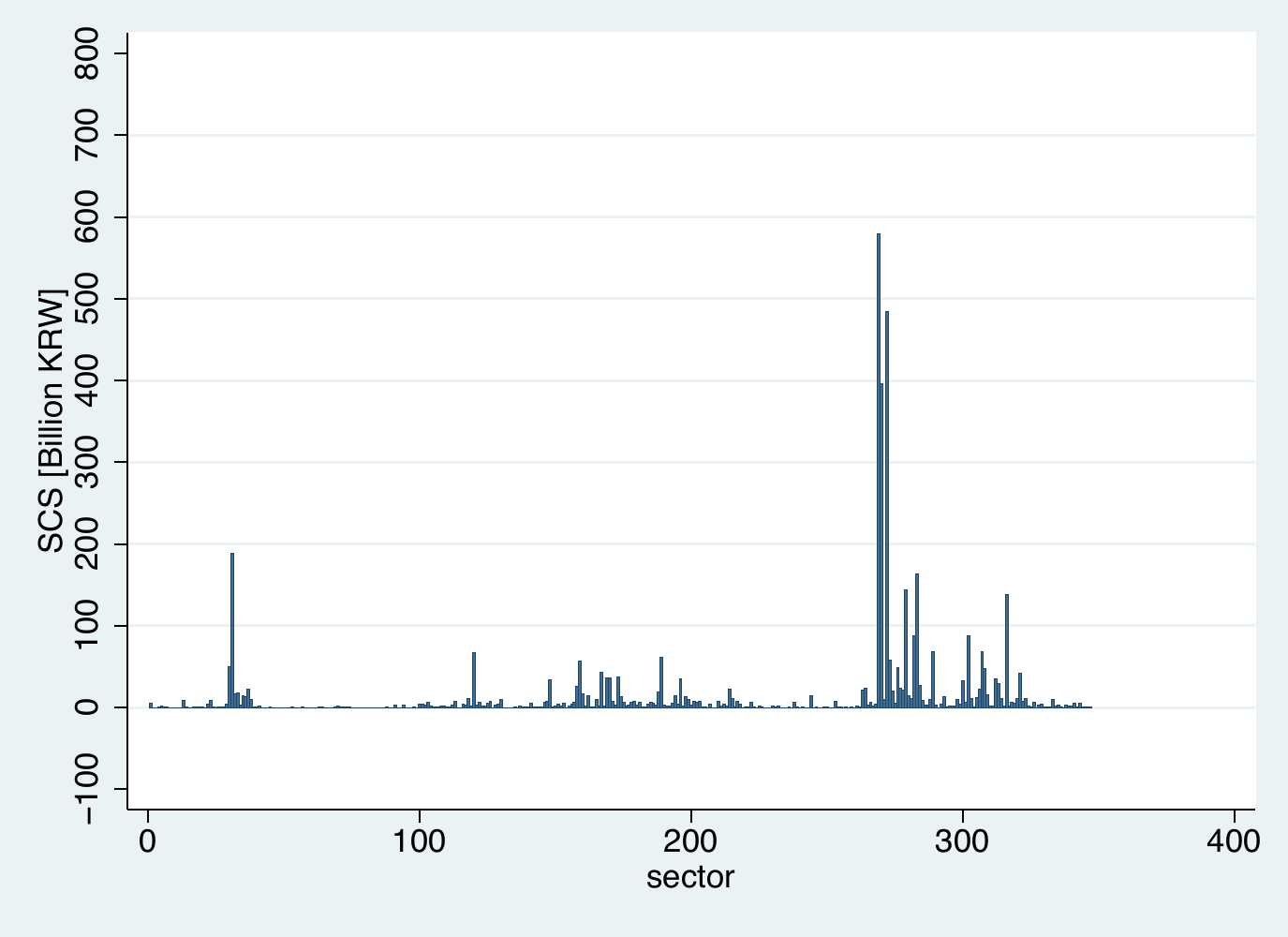}
 \caption{Sectoral distribution of SCS for productivity doubling of RMC sector (159th) for Cobb--Douglas system. (Korea)} \label{scsdCDK}
\end{figure}

\section{Concluding Remarks}
To date input--output analysis has been a one-of-a-kind framework that considers industry-wide propagation when assessing the costs and benefits of new goods and  innovations.  
Input--output analysis, nonetheless, has laid its theory upon the non-substitution theorem, which allows the researcher to study under a fixed technological structure while restricting the subjects of analyses to transformations within the final demand.  
Substitution of technology will nevertheless take place in any industry when a new technology is actually introduced into any component (industry) of the economy.
Larger influence is typically foreseeable for intermediate industries, as they have much larger and wider feedback on economy-wide systems of production.

In order to consider all technology substitution possibilities, we proposed in this study a methodology to measure the sector specific substitution elasticity for the CES production function, rather than using uniform \text{a priori} substitution elasticity (such as zeros and ones), when modeling the economy-wide multisector multifactor production system.
A dual analytical method (i.e., unit cost functions) was used to evaluate influences upon general equilibrium technological substitutions and eventually upon social costs and benefits, initiated by the introduction of innovation, which we treat as gains in productivity.
We have found that the more elastic production functions (Cobb--Douglas and CES) have more significant and wider propagation effects, whereas inelastic production functions (Leontief) have effects that are relatively less and polarized. Applications and extensions of this framework can perhaps be immense, including internationalization, dynamicalization, and quality considerations all remaining for future investigations.

\begin{figure}[t!]
\centering
  \includegraphics[scale = 0.63]{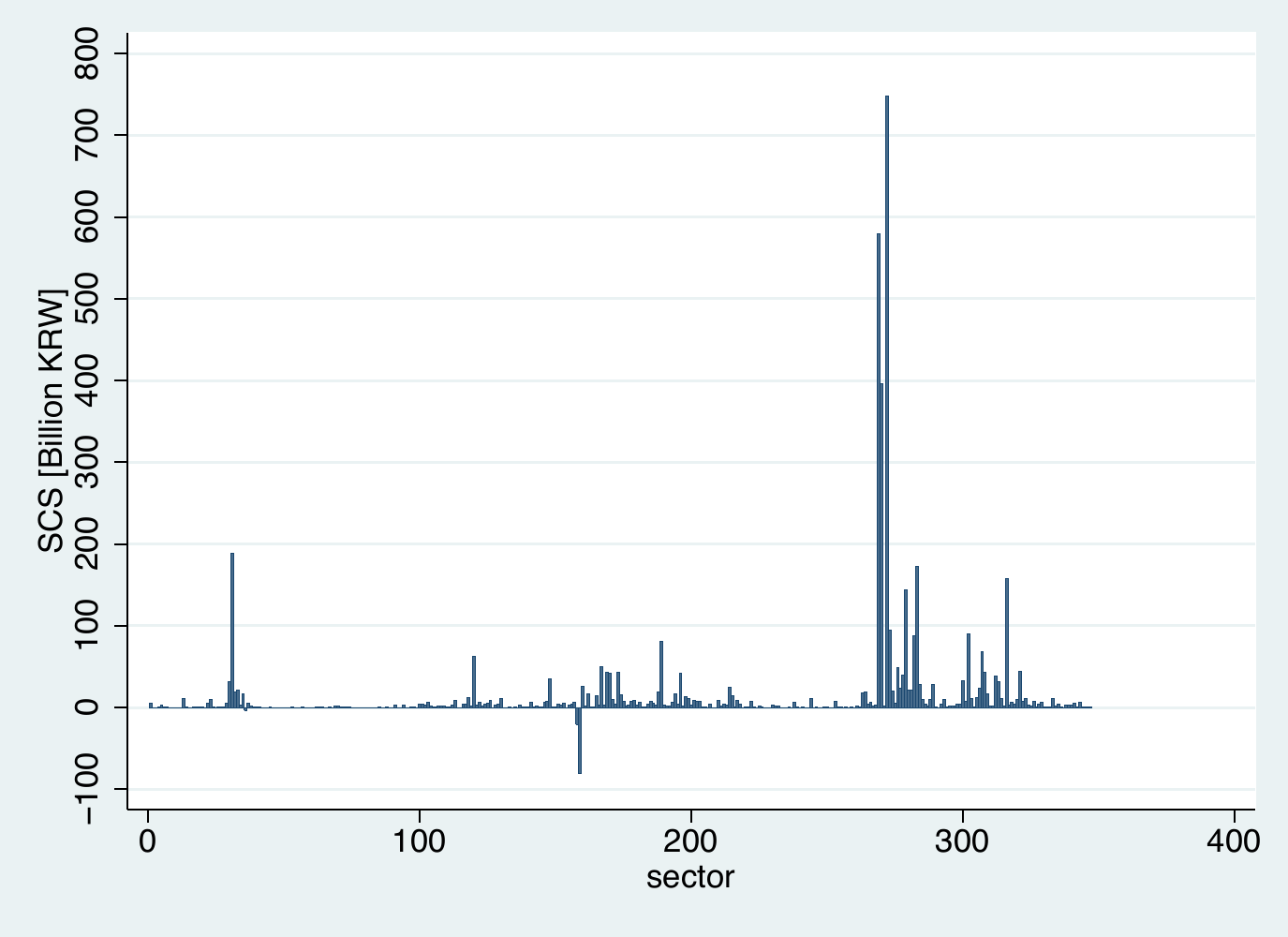}
 \caption{Sectoral distribution of SCS for productivity doubling of RMC sector (159th) for CES system. (Korea)} \label{scsdCESK}
 \end{figure}

 \begin{figure}[t!]
\centering
  \includegraphics[scale = 0.63]{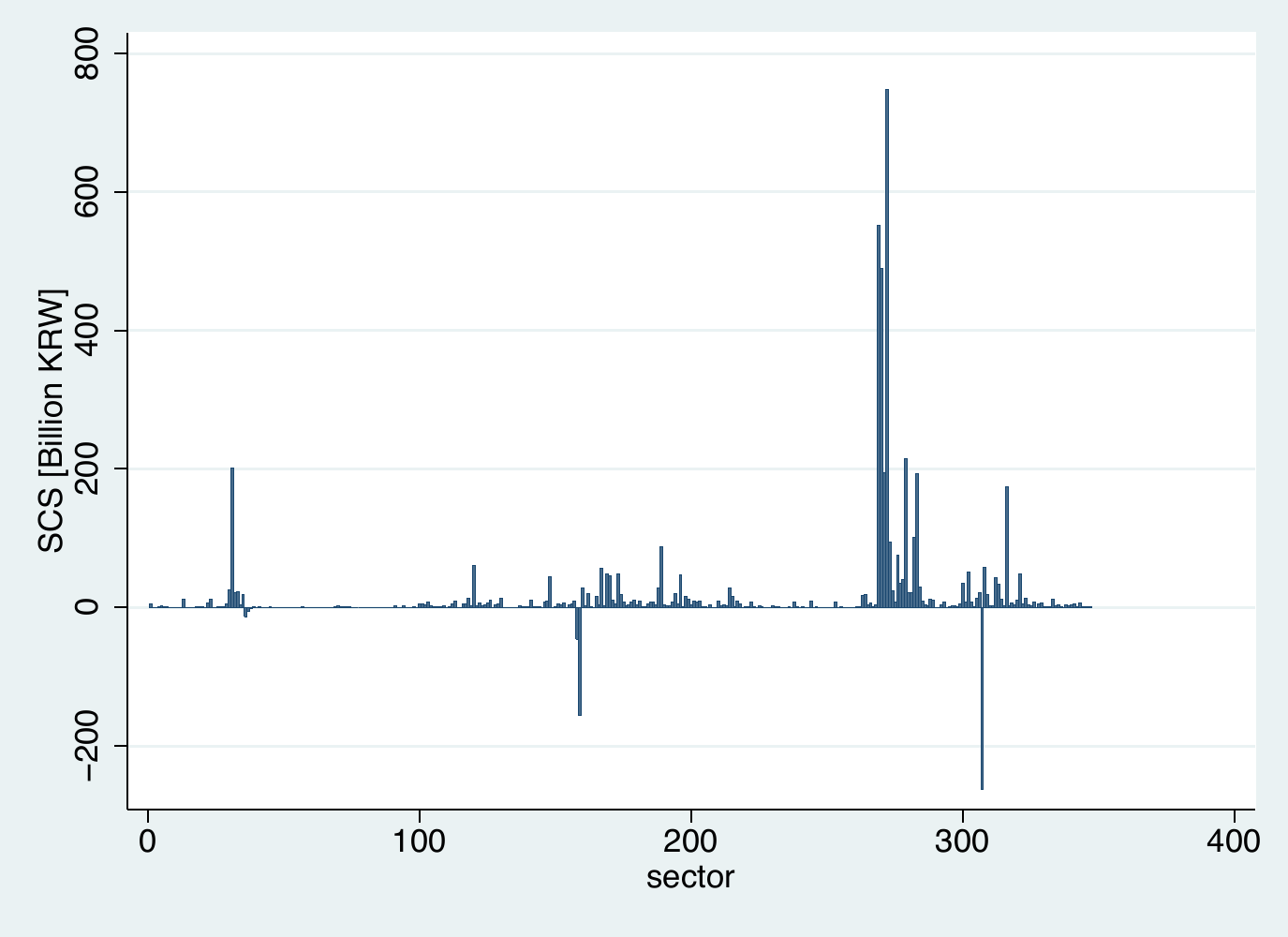}
 \caption{Sectoral distribution of SCS for productivity doubling of RMC sector (159th) for multifactor CES (all estimates) system. (Korea)} \label{scsdCESrK}
 \end{figure}

\section*{Acknowledgements}
The authors would like to thank the anonymous reviewers and the editor of this journal for helpful comments and suggestions.
This material is based upon work supported by the Japan Society for the Promotion of Science under Grant No.16K00687.

{\raggedright
\bibliography{netinno_bib}
}

\clearpage
\section*{Appendix}
\setcounter{table}{0}
\renewcommand{\thetable}{A\arabic{table}}
\newcolumntype{.}{D{.}{.}{0}}
{\tiny
\tablecaption{CES Elasticities and Productivity Growths (Japan 2000--2005)}
\label{tab_JPN}
\tablefirsthead{
\hline  
\multicolumn{1}{c}{sector}  
& \multicolumn{2}{c}{Elasticity}  
& \multicolumn{2}{c}{TFPg}    
& \multicolumn{1}{c}{Obs.} \\ 
\hline\noalign{\smallskip}
}
\tablehead{\multicolumn{6}{c}{Table Continued} \\
\hline
\multicolumn{1}{c}{sector}  
& \multicolumn{2}{c}{Elasticity}  
& \multicolumn{2}{c}{TFPg}    
& \multicolumn{1}{c}{Obs.} \\ 
\hline\noalign{\smallskip}
}
\tabletail{\hline}
\begin{xtabular}{l.l.l.}
Liquid crystal element	&	2.296	&	***	&	1.269	&	***		(***)	&	116	\\
Turbines	&	1.689	&	***	&	0.783	&	***		(***)	&	119	\\
Video recording and playback equipment	&	2.007	&	***	&	0.773	&	***		(***)	&	136	\\
Personal Computers	&	1.455	&	*	&	0.647	&				&	126	\\
Coal products	&	1.979	&	**	&	0.593	&		(***)	&	91	\\
Frozen fish and shellfish	&	2.074	&	*	&	0.449	&			(***)	&	80	\\
Electronic computing equipment (accessory equipment)	&	1.871	&	***	&	0.412	&	**		(***)	&	132	\\
Cyclic intermediates	&	1.784	&	***	&	0.367	&			(***)	&	105	\\
Fowls and broilers	&	2.199	&	*	&	0.332	&			(***)	&	57	\\
Steel ships	&	1.451	&	***	&	0.307	&	**		(***)	&	157	\\
Photographic sensitive materials	&	1.581	&	**	&	0.283	&			(**)	&	106	\\
Other business services	&	2.098	&	***	&	0.270	&			(***)	&	122	\\
Electronic computing equipment (except personal computers)	&	1.668	&	***	&	0.268	&				&	126	\\
Financial service	&	0.275	&	*	&	0.260	&			(***)	&	101	\\
Social welfare (profit-making)	&	1.268	&	***	&	0.251	&			(***)	&	143	\\
Private non-profit institutions serving households, n.e.c. *	&	1.391	&	*	&	0.242	&			(***)	&	105	\\
Repair of ships	&	1.378	&	**	&	0.239	&			(***)	&	142	\\
Inorganic pigment	&	1.581	&	**	&	0.233	&			(***)	&	104	\\
Other iron or steel products	&	1.345	&	*	&	0.231	&				&	81	\\
Public administration (central) **	&	1.603	&	***	&	0.223	&			(***)	&	219	\\
Boilers	&	1.646	&	**	&	0.217	&			(***)	&	120	\\
Aliphatic intermediates	&	1.461	&	*	&	0.214	&			(**)	&	109	\\
Household electric appliances (except air-conditioners)	&	1.333	&	**	&	0.182	&				&	153	\\
Medical service (medical corporations, etc.)	&	1.622	&	**	&	0.168	&			(**)	&	156	\\
Synthetic dyes	&	1.868	&	***	&	0.165	&			(***)	&	97	\\
Dishes, sushi and lunch boxes	&	1.761	&	**	&	0.165	&			(***)	&	116	\\
Applied electronic equipment	&	1.455	&	**	&	0.160	&			(*)	&	133	\\
Railway transport (freight)	&	1.918	&	***	&	0.154	&			(***)	&	101	\\
Noodles	&	1.669	&	**	&	0.151	&			(*)	&	108	\\
Motor vehicle parts and accessories	&	1.701	&	***	&	0.137	&	*		(**)	&	152	\\
Dextrose, syrup and isomerized sugar	&	1.405	&	**	&	0.133	&			(**)	&	78	\\
Medicaments	&	1.976	&	*	&	0.132	&				&	135	\\
Electric bulbs	&	1.570	&	**	&	0.125	&			(*)	&	103	\\
Other electrical devices and parts	&	2.059	&	***	&	0.121	&				&	125	\\
Other general industrial machinery and equipment	&	1.386	&	*	&	0.116	&				&	140	\\
Other industrial organic chemicals	&	1.687	&	*	&	0.115	&				&	118	\\
Metal containers, fabricated plate and sheet metal	&	1.780	&	***	&	0.104	&	**		(**)	&	134	\\
Metallic furniture and fixture	&	1.775	&	**	&	0.103	&				&	124	\\
Nursing care (In-facility)	&	1.585	&	***	&	0.101	&			(**)	&	159	\\
Semiconductor making equipment	&	1.453	&	**	&	0.099	&				&	142	\\
Marine culture	&	1.717	&	**	&	0.092	&				&	92	\\
Other metal products	&	1.774	&	***	&	0.087	&			(*)	&	145	\\
Bearings	&	1.627	&	***	&	0.086	&				&	114	\\
Pumps and compressors	&	2.111	&	***	&	0.085	&			(**)	&	129	\\
Wheat, barley and the like	&	2.952	&	*	&	0.081	&				&	60	\\
Confectionery	&	1.807	&	***	&	0.080	&				&	121	\\
Other educational and training institutions (profit-making)	&	1.748	&	**	&	0.079	&				&	74	\\
Sporting and athletic goods	&	1.578	&	**	&	0.077	&				&	135	\\
Cosmetics, toilet preparations and dentifrices	&	1.576	&	*	&	0.074	&				&	105	\\
Tires and inner tubes	&	1.517	&	*	&	0.072	&				&	102	\\
Miscellaneous manufacturing products	&	1.622	&	***	&	0.071	&				&	180	\\
Gas and oil appliances and heating and cooking apparatus	&	1.568	&	***	&	0.069	&				&	133	\\
Agricultural public construction	&	2.039	&	*	&	0.062	&				&	144	\\
Health and hygiene (profit-making)	&	1.509	&	**	&	0.059	&				&	94	\\
Plumber's supplies, powder metallurgy products and tools	&	1.596	&	***	&	0.057	&				&	128	\\
Internal combustion engines for vessels	&	1.808	&	**	&	0.057	&				&	115	\\
Other rubber products	&	1.740	&	***	&	0.052	&				&	125	\\
Electric wires and cables	&	1.566	&	***	&	0.051	&				&	121	\\
Other final chemical products	&	1.782	&	***	&	0.048	&				&	150	\\
Activities not elsewhere classified	&	3.575	&	***	&	0.047	&				&	179	\\
Paint and varnishes	&	1.703	&	***	&	0.047	&				&	125	\\
Oil and fat industrial chemicals	&	1.555	&	*	&	0.047	&				&	91	\\
Compressed gas and liquefied gas	&	1.593	&	*	&	0.041	&				&	81	\\
Metal products for construction	&	1.497	&	**	&	0.040	&				&	136	\\
Other pulp, paper and processed paper products	&	1.517	&	**	&	0.035	&				&	125	\\
Metal molds	&	1.894	&	***	&	0.035	&				&	127	\\
Health and hygiene (public) **	&	1.496	&	***	&	0.033	&				&	91	\\
Machinery for agricultural use	&	1.576	&	**	&	0.030	&				&	142	\\
Publication	&	1.470	&	*	&	0.029	&				&	105	\\
Other special machinery for industrial use 	&	1.646	&	**	&	0.026	&				&	146	\\
Other industrial inorganic chemicals	&	1.643	&	**	&	0.026	&				&	116	\\
Abrasive	&	1.363	&	*	&	0.025	&				&	126	\\
Other services relating to communication	&	2.444	&	***	&	0.019	&				&	65	\\
Advertising services	&	1.964	&	***	&	0.018	&				&	103	\\
Electron tubes	&	1.825	&	***	&	0.018	&				&	116	\\
Retort foods	&	1.543	&	*	&	0.012	&				&	92	\\
Chemical fertilizer	&	1.608	&	*	&	0.012	&				&	113	\\
Internal combustion engines for motor vehicles and parts	&	1.803	&	***	&	0.010	&				&	131	\\
Other structural clay products	&	1.485	&	**	&	0.010	&				&	107	\\
Newspaper	&	1.529	&	**	&	0.007	&				&	99	\\
Wooden furniture and fixtures	&	2.086	&	***	&	0.004	&				&	145	\\
Coated steel	&	1.981	&	***	&	0.004	&				&	100	\\
Miscellaneous ceramic, stone and clay products	&	1.455	&	***	&	0.004	&				&	147	\\
Cement	&	1.577	&	**	&	0.000	&				&	103	\\
Glass fiber and glass fiber products, n.e.c.	&	1.774	&	***	&	-0.002	&				&	106	\\
Conveyors	&	1.408	&	**	&	-0.005	&				&	138	\\
Fisheries	&	1.648	&	***	&	-0.011	&				&	92	\\
Other general machines and parts	&	1.644	&	***	&	-0.013	&				&	143	\\
Sewage disposal **	&	1.734	&	***	&	-0.013	&				&	86	\\
Other photographic and optical instruments	&	0.423	&	**	&	-0.014	&				&	127	\\
Bread	&	1.664	&	**	&	-0.015	&				&	111	\\
Office supplies	&	2.608	&	***	&	-0.015	&				&	29	\\
Wiring devices and supplies	&	1.784	&	***	&	-0.019	&				&	128	\\
Electrical equipment for internal combustion engines	&	1.483	&	**	&	-0.021	&				&	130	\\
Medical service (non-profit foundations, etc.)	&	1.812	&	***	&	-0.021	&				&	154	\\
Clay refractories	&	1.656	&	***	&	-0.022	&				&	109	\\
Cast and forged materials (iron)	&	2.091	&	***	&	-0.026	&				&	133	\\
Engines	&	1.859	&	***	&	-0.026	&				&	129	\\
Pulp	&	2.634	&	**	&	-0.028	&				&	104	\\
Non-ferrous metal castings and forgings	&	1.615	&	**	&	-0.034	&				&	123	\\
Other wooden products	&	1.716	&	***	&	-0.035	&				&	160	\\
Railway transport (passengers)	&	2.086	&	***	&	-0.040	&				&	112	\\
Sugar	&	1.492	&	**	&	-0.044	&				&	83	\\
News syndicates and private detective agencies	&	1.434	&	*	&	-0.045	&				&	74	\\
Other electronic components	&	1.746	&	***	&	-0.049	&				&	152	\\
Electricity	&	1.476	&	*	&	-0.052	&				&	98	\\
Medical instruments	&	0.090	&	***	&	-0.052	&				&	151	\\
Repair of motor vehicles	&	1.442	&	*	&	-0.052	&				&	114	\\
Repair of rolling stock	&	1.712	&	***	&	-0.052	&				&	117	\\
Other glass products	&	2.006	&	***	&	-0.060	&				&	107	\\
Bolts, nuts, rivets and springs	&	1.763	&	***	&	-0.060	&				&	132	\\
Rolled and drawn aluminum	&	1.824	&	*	&	-0.063	&				&	86	\\
Synthetic fibers	&	1.636	&	*	&	-0.065	&				&	99	\\
Woven fabric apparel	&	1.577	&	*	&	-0.065	&				&	101	\\
Whiskey and brandy	&	2.601	&	*	&	-0.071	&				&	88	\\
Social welfare (private, non-profit) *	&	1.460	&	***	&	-0.072	&				&	143	\\
Knitted apparel	&	2.031	&	*	&	-0.084	&				&	107	\\
Accommodations	&	1.825	&	***	&	-0.084	&			(**)	&	161	\\
Medical service (public)	&	1.808	&	***	&	-0.087	&			(**)	&	153	\\
Other transport equipment	&	1.973	&	***	&	-0.089	&				&	140	\\
Pottery, china and earthenware	&	2.073	&	***	&	-0.089	&			(*)	&	119	\\
Fiber yarns	&	1.851	&	**	&	-0.094	&				&	94	\\
Plastic footwear	&	1.965	&	***	&	-0.095	&	**		(**)	&	108	\\
Nursing care (In-home)	&	1.552	&	***	&	-0.095	&			(**)	&	153	\\
Transformers and reactors	&	1.600	&	**	&	-0.102	&				&	124	\\
Cast iron pipes and tubes	&	1.805	&	**	&	-0.102	&				&	90	\\
Cleaning	&	1.655	&	**	&	-0.103	&			(*)	&	88	\\
Aircrafts	&	1.684	&	**	&	-0.103	&				&	121	\\
Food processing machinery and equipment	&	1.562	&	**	&	-0.116	&			(*)	&	124	\\
Industrial robots	&	1.520	&	**	&	-0.117	&				&	124	\\
Beauty shops	&	1.459	&	*	&	-0.126	&				&	91	\\
Plywood	&	1.713	&	**	&	-0.126	&	*			&	86	\\
Passenger motor cars	&	1.703	&	**	&	-0.135	&			(*)	&	123	\\
Audio and video records, other information recording media	&	1.488	&	*	&	-0.135	&			(*)	&	95	\\
Motor vehicle bodies	&	1.592	&	*	&	-0.139	&				&	125	\\
Barber shops	&	1.657	&	***	&	-0.148	&			(***)	&	86	\\
Repair of machine	&	1.622	&	**	&	-0.153	&			(*)	&	145	\\
Plasticizers	&	2.262	&	***	&	-0.153	&	*		(***)	&	84	\\
Other personal services	&	1.925	&	*	&	-0.155	&			(**)	&	113	\\
Rolled and drawn copper and copper alloys	&	1.829	&	**	&	-0.166	&				&	83	\\
Textile machinery	&	2.218	&	***	&	-0.169	&	*		(***)	&	138	\\
Rotating electrical equipment	&	1.457	&	**	&	-0.172	&			(**)	&	127	\\
Chemical machinery	&	1.528	&	**	&	-0.176	&			(**)	&	132	\\
Public baths	&	1.544	&	*	&	-0.188	&			(**)	&	94	\\
Metal processing machinery	&	1.654	&	***	&	-0.192	&	**		(***)	&	128	\\
Petrochemical basic products	&	1.798	&	*	&	-0.200	&				&	89	\\
Image information production and distribution industry	&	1.678	&	**	&	-0.201	&			(**)	&	119	\\
Social welfare (public) **	&	1.479	&	**	&	-0.201	&			(***)	&	142	\\
Hot rolled steel	&	2.138	&	***	&	-0.207	&				&	97	\\
Crops for feed and forage	&	2.988	&	***	&	-0.207	&	***		(***)	&	58	\\
Crude steel (electric furnaces)	&	1.870	&	**	&	-0.226	&				&	96	\\
Machinery for service industry	&	1.378	&	**	&	-0.233	&			(**)	&	129	\\
Social education (public) **	&	1.812	&	*	&	-0.238	&			(***)	&	93	\\
Consigned freight forwarding	&	-0.732	&	**	&	-0.239	&			(*)	&	93	\\
Wired communication equipment	&	2.164	&	***	&	-0.243	&	*		(***)	&	150	\\
Other electrical devices and parts	&	1.388	&	**	&	-0.246	&			(***)	&	142	\\
Iron and steel shearing and slitting	&	2.379	&	***	&	-0.265	&			(*)	&	83	\\
Other wearing apparel and clothing accessories	&	1.800	&	*	&	-0.270	&			(***)	&	109	\\
Coal mining , crude petroleum and natural gas	&	1.850	&	***	&	-0.277	&	***		(***)	&	89	\\
Rolling stock	&	1.808	&	***	&	-0.284	&	***		(***)	&	138	\\
Research and development (intra-enterprise)	&	1.461	&	**	&	-0.317	&	*		(***)	&	126	\\
Batteries	&	1.640	&	**	&	-0.317	&			(***)	&	129	\\
Watches and clocks	&	1.471	&	***	&	-0.339	&			(***)	&	121	\\
Wooden chips	&	1.626	&	*	&	-0.350	&			(***)	&	64	\\
Optical fiber cables	&	1.634	&	**	&	-0.360	&			(***)	&	115	\\
Crude steel (converters)	&	2.635	&	***	&	-0.377	&	**		(***)	&	99	\\
Electric measuring instruments	&	1.362	&	*	&	-0.399	&			(***)	&	128	\\
Storage facility service	&	1.602	&	**	&	-0.404	&			(***)	&	105	\\
Copper	&	2.110	&	**	&	-0.448	&				&	77	\\
Private non-profit institutions serving enterprises	&	1.586	&	*	&	-0.450	&			(***)	&	91	\\
Other non-ferrous metal products	&	2.152	&	**	&	-0.549	&	***		(**)	&	88	\\
Pig iron	&	1.600	&	**	&	-0.680	&	*		(*)	&	169	\\
Research institutes for natural science (pubic) **	&	2.090	&	*	&	-0.745	&	*		(***)	&	90	\\
Metallic ores	&	1.634	&	***	&	-0.799	&	***		(***)	&	82	\\
Ferro alloys	&	1.652	&	*	&	-0.823	&				&	85	\\
Research institutes for natural sciences (profit-making)	&	2.108	&	**	&	-0.855	&			(***)	&	93	\\
\hline
\end{xtabular}
\\\\
\normalsize
Note: The statistical significances in parenthesis are of the intercept of the regression (\ref{main}).
}

\clearpage
{\tiny
\tablecaption{CES Elasticities and Productivity Growths (Korea 2000--2005)}
\label{tab_KOR}
\tablefirsthead{
\hline
\multicolumn{1}{c}{sector}  
& \multicolumn{2}{c}{Elasticity}  
& \multicolumn{2}{c}{TFPg}    
& \multicolumn{1}{c}{Obs.} \\ 
\hline\noalign{\smallskip}
}
\tablehead{\multicolumn{6}{c}{Table Continued} \\
\hline
\multicolumn{1}{c}{sector}  
& \multicolumn{2}{c}{Elasticity}  
& \multicolumn{2}{c}{TFPg}  
& \multicolumn{1}{c}{Obs.} \\ 
\hline\noalign{\smallskip}
}
\tabletail{\hline}
\begin{xtabular}{l.l.l.}
Photographic and optical instruments	&	2.116	&	***	&	0.688	&	***		(***)	&	165	\\
Computer and peripheral equipment	&	1.660	&	**	&	0.619	&			(*)	&	166	\\
Watches and clocks	&	1.615	&	**	&	0.618	&			(***)	&	147	\\
Electric resistors and storage batteries	&	2.033	&	***	&	0.582	&	***		(***)	&	156	\\
Research institutes(private, non-profit, commercial)	&	1.498	&	*	&	0.566	&			(***)	&	152	\\
Electric household audio equipment	&	2.141	&	***	&	0.564	&	*		(***)	&	151	\\
Misc. amusement and recreation services	&	1.817	&	***	&	0.514	&	***		(***)	&	153	\\
Supporting land transport activities	&	1.555	&	**	&	0.507	&			(***)	&	126	\\
Wood furniture	&	1.495	&	*	&	0.447	&			(***)	&	165	\\
Education (commercial)	&	1.682	&	**	&	0.417	&			(***)	&	127	\\
Other audio and visual equipment	&	1.614	&	*	&	0.402	&			(*)	&	164	\\
Bicycles and parts and misc. transportation equipment	&	1.860	&	***	&	0.400	&	***		(***)	&	132	\\
Household laundry equipment	&	1.480	&	**	&	0.399	&			(***)	&	145	\\
Electron tubes	&	1.709	&	***	&	0.393	&	***		(**)	&	159	\\
Semiconductor devices	&	1.542	&	**	&	0.371	&				&	162	\\
Road freight transport	&	1.961	&	**	&	0.370	&			(***)	&	131	\\
Printed circuit boards	&	1.550	&	**	&	0.357	&	*		(*)	&	160	\\
Section steel	&	1.520	&	**	&	0.340	&			(*)	&	121	\\
Supporting air transport activities	&	2.164	&	***	&	0.339	&	**		(***)	&	108	\\
Business and professional organizations	&	2.735	&	***	&	0.335	&			(***)	&	95	\\
Passenger automobiles	&	1.674	&	***	&	0.334	&	**		(***)	&	155	\\
Office machines and devices	&	1.536	&	*	&	0.332	&			(**)	&	154	\\
Industrial glass products	&	2.121	&	***	&	0.292	&	***		(**)	&	169	\\
Central bank and banking institutions, Non-bank depository institutions	&	1.864	&	**	&	0.287	&			(***)	&	120	\\
Water supply	&	1.675	&	**	&	0.285	&			(**)	&	124	\\
Road passenger transport	&	1.983	&	***	&	0.285	&			(**)	&	131	\\
Clay products for construction	&	1.800	&	**	&	0.285	&			(**)	&	140	\\
Lime, gypsum, and plaster products	&	1.813	&	*	&	0.282	&			(***)	&	134	\\
Food processing machinery	&	1.592	&	**	&	0.278	&	**		(***)	&	143	\\
Boiler, Heating apparatus and cooking appliances	&	1.610	&	*	&	0.274	&			(**)	&	164	\\
Pulp	&	1.526	&	*	&	0.273	&				&	112	\\
Medical instruments and supplies	&	1.793	&	***	&	0.271	&			(**)	&	167	\\
Regulators and Measuring and analytical instruments	&	1.603	&	**	&	0.266	&			(**)	&	167	\\
Coastal and inland water transport	&	1.552	&	**	&	0.265	&			(***)	&	134	\\
Leather	&	1.831	&	**	&	0.260	&	*		(**)	&	129	\\
Cosmetics and dentifrices	&	1.974	&	**	&	0.255	&	*		(**)	&	165	\\
Non-life insurance	&	1.586	&	*	&	0.250	&			(*)	&	107	\\
Misc. chemical products	&	1.589	&	**	&	0.245	&			(**)	&	172	\\
Sports organizations and sports facility operation	&	1.635	&	***	&	0.241	&			(**)	&	144	\\
Social work activities(other)	&	1.757	&	**	&	0.229	&			(**)	&	137	\\
Trucks and Motor vehicles with special equipment	&	1.845	&	***	&	0.229	&	***		(***)	&	154	\\
Other membership organizations	&	1.855	&	**	&	0.225	&			(**)	&	114	\\
Wooden containers and Other wooden products	&	2.034	&	**	&	0.224	&			(**)	&	124	\\
Bakery and confectionery products	&	1.819	&	*	&	0.213	&			(**)	&	174	\\
Household refrigerators	&	1.795	&	***	&	0.213	&	**		(***)	&	152	\\
Asbestos and mineral wool products	&	1.754	&	**	&	0.212	&			(**)	&	145	\\
Air-conditioning equipment and industrial refrigeration equipment	&	1.524	&	**	&	0.209	&				&	163	\\
Buses and vans	&	1.736	&	***	&	0.208	&			(***)	&	152	\\
Medicaments	&	1.998	&	***	&	0.207	&	***		(**)	&	175	\\
Textile machinery	&	1.468	&	*	&	0.199	&				&	165	\\
Silk and hempen fabrics	&	1.982	&	**	&	0.196	&	*			&	110	\\
Printing ink	&	2.049	&	***	&	0.190	&	**		(***)	&	127	\\
Motors and generators	&	1.731	&	***	&	0.187	&	**		(**)	&	161	\\
Misc. non-metallic minerals	&	2.262	&	***	&	0.185	&				&	108	\\
Sanitary services(public)	&	1.701	&	**	&	0.185	&				&	130	\\
Concrete blocks, bricks, and other concrete products	&	1.891	&	***	&	0.182	&	**		(***)	&	144	\\
Lubricants	&	1.736	&	*	&	0.180	&				&	131	\\
Pottery	&	1.560	&	*	&	0.177	&				&	155	\\
Railroad vehicles and parts	&	1.537	&	**	&	0.174	&				&	157	\\
Metal molds and industrial patterns	&	1.662	&	**	&	0.169	&				&	152	\\
Luggage and handbags	&	2.172	&	***	&	0.161	&	**		(***)	&	118	\\
Pens, pencils, and other artists' materials	&	1.794	&	***	&	0.160	&	**		(*)	&	145	\\
Motion picture, Theatrical producers, bands, and entertainers	&	1.619	&	***	&	0.158	&			(*)	&	151	\\
Dairy products	&	1.971	&	**	&	0.157	&			(*)	&	144	\\
Publishing	&	1.473	&	*	&	0.154	&				&	124	\\
Ship repairing and ship parts	&	1.799	&	***	&	0.154	&	**		(**)	&	151	\\
Misc. nonmetallic minerals products	&	1.680	&	*	&	0.152	&				&	140	\\
Household glass products and others	&	1.940	&	***	&	0.143	&	**		(*)	&	136	\\
Agricultural implements and machinery	&	1.620	&	***	&	0.129	&				&	155	\\
Social work activities(public)	&	2.169	&	***	&	0.124	&				&	121	\\
Reproduction of recorded media	&	1.987	&	***	&	0.123	&	*		(**)	&	136	\\
Anthracite	&	2.325	&	***	&	0.122	&				&	132	\\
Paints, varnishes, and allied products	&	1.700	&	**	&	0.118	&				&	155	\\
Line telecommunication apparatuses	&	1.636	&	**	&	0.118	&				&	161	\\
Leather wearing apparels	&	1.845	&	*	&	0.116	&				&	108	\\
Library, museum and similar recreation related services(public)	&	1.843	&	***	&	0.112	&				&	133	\\
Paper containers	&	1.927	&	***	&	0.107	&				&	132	\\
Knitted clothing accessories	&	2.204	&	**	&	0.100	&				&	116	\\
Synthetic fiber fabrics	&	1.852	&	**	&	0.097	&				&	128	\\
Motorcycles and parts	&	1.687	&	**	&	0.095	&				&	148	\\
Accommodation	&	1.657	&	**	&	0.094	&				&	132	\\
Ginseng products	&	1.686	&	*	&	0.089	&				&	104	\\
Sheet glass and primary glass products	&	1.985	&	***	&	0.088	&				&	129	\\
Electric transformers	&	1.851	&	***	&	0.087	&				&	150	\\
Salted, dried and smoked seafoods	&	3.290	&	*	&	0.084	&				&	98	\\
Misc. electric equipment and supplies	&	1.503	&	*	&	0.082	&				&	155	\\
Printing	&	1.579	&	***	&	0.081	&				&	143	\\
Abrasives	&	1.710	&	**	&	0.074	&				&	142	\\
Cement	&	2.086	&	***	&	0.070	&				&	154	\\
Prepared livestock feeds	&	1.713	&	*	&	0.069	&				&	154	\\
Library, museum and similar recreation related services(other)	&	1.578	&	*	&	0.066	&				&	135	\\
Knitted fabrics	&	1.928	&	**	&	0.064	&				&	111	\\
Internal combustion engines and turbines	&	1.649	&	***	&	0.063	&				&	156	\\
Fiber bleaching and dyeing	&	1.949	&	**	&	0.058	&				&	119	\\
Cleaning and disinfection services	&	1.552	&	*	&	0.058	&				&	104	\\
Other paper products	&	1.597	&	*	&	0.054	&				&	160	\\
Other raw paper and paperboard	&	1.808	&	***	&	0.043	&				&	150	\\
Petrochemical intermediate products and Other basic organic chemicals	&	1.876	&	**	&	0.042	&				&	163	\\
Fastening metal products	&	1.661	&	**	&	0.038	&				&	137	\\
Household articles of plastic material	&	1.721	&	**	&	0.032	&				&	124	\\
Stationery paper and office paper	&	1.497	&	*	&	0.032	&				&	125	\\
Recording media and Photographic chemical products	&	1.853	&	***	&	0.031	&				&	142	\\
Medical and health services (commercial)	&	2.288	&	***	&	0.030	&				&	160	\\
Ready mixed concrete	&	2.040	&	***	&	0.030	&				&	132	\\
Supporting water transport activities	&	1.637	&	**	&	0.029	&				&	125	\\
Other leather products	&	1.858	&	*	&	0.028	&				&	91	\\
Construction and mining machinery	&	1.577	&	**	&	0.025	&				&	156	\\
Nitrogen compounds	&	1.759	&	**	&	0.025	&				&	114	\\
Road construction	&	1.389	&	*	&	0.023	&				&	179	\\
Metal products for construction	&	1.828	&	**	&	0.019	&				&	134	\\
Industrial plastic products	&	1.674	&	**	&	0.014	&				&	167	\\
Land clearing and reclamation, and irrigation project construction	&	1.539	&	**	&	0.009	&				&	167	\\
Soy sauce ad bean paste	&	1.750	&	*	&	0.008	&				&	127	\\
Communications line construction	&	1.585	&	**	&	0.006	&				&	159	\\
Metal furniture	&	1.565	&	**	&	0.006	&				&	146	\\
Thread and other fiber yarns	&	1.915	&	***	&	-0.004	&				&	114	\\
Life insurance	&	1.627	&	*	&	-0.005	&				&	106	\\
Capacitors and rectifiers, Electric transmission and distribution equipment	&	1.583	&	***	&	-0.005	&				&	167	\\
Musical instruments	&	1.506	&	**	&	-0.005	&				&	155	\\
Iron foundries and foundry iron pipe and tubes	&	1.840	&	***	&	-0.006	&				&	152	\\
Misc. petroleum refinery products	&	1.793	&	*	&	-0.011	&				&	127	\\
Medical and health services(public)	&	2.180	&	***	&	-0.011	&				&	138	\\
Pumps and compressors	&	1.601	&	**	&	-0.018	&				&	158	\\
Adhesives, gelatin and sealants	&	1.882	&	**	&	-0.021	&				&	143	\\
Rubber products	&	1.763	&	***	&	-0.022	&				&	154	\\
Canned or cured fruits and vegetables	&	1.761	&	*	&	-0.034	&				&	139	\\
Corrugated paper and solid fiber boxes	&	1.662	&	**	&	-0.040	&				&	119	\\
Crushed and broken stone abd Other bulk stones	&	1.787	&	*	&	-0.044	&				&	120	\\
Railroad construction	&	1.432	&	*	&	-0.045	&				&	170	\\
Medical and health services(Private, non-profit)	&	1.946	&	***	&	-0.046	&				&	141	\\
Architectural engineering services	&	1.606	&	**	&	-0.048	&				&	143	\\
Newspapers	&	1.873	&	***	&	-0.049	&				&	118	\\
Sporting and athletic goods	&	1.720	&	*	&	-0.058	&				&	159	\\
Treatment and coating of metals and Misc. fabricated metal products	&	1.722	&	**	&	-0.060	&				&	171	\\
Synthetic fiber yarn	&	1.903	&	**	&	-0.067	&				&	124	\\
Plywood	&	1.769	&	*	&	-0.067	&				&	122	\\
Electric lamps and electric lighting fixtures	&	1.575	&	**	&	-0.068	&				&	160	\\
Synthetic fibers	&	1.701	&	*	&	-0.073	&				&	128	\\
Research institutes(public)	&	1.611	&	**	&	-0.080	&				&	182	\\
Services related to real estate	&	2.091	&	**	&	-0.080	&				&	91	\\
Lumber	&	2.081	&	**	&	-0.080	&				&	105	\\
Insulated wires and cables	&	1.777	&	***	&	-0.089	&				&	169	\\
Other nonferrous metal ingots	&	1.697	&	*	&	-0.097	&				&	121	\\
Personal services	&	1.977	&	***	&	-0.110	&			(*)	&	124	\\
Conveyors and conveying equipment	&	1.649	&	**	&	-0.110	&				&	165	\\
Electric power plant construction	&	1.334	&	*	&	-0.125	&				&	171	\\
Starches	&	2.220	&	**	&	-0.137	&			(*)	&	102	\\
Footwear	&	1.836	&	***	&	-0.139	&	*		(*)	&	131	\\
Other edible crops	&	2.586	&	**	&	-0.139	&				&	58	\\
Explosives and fireworks products	&	1.637	&	**	&	-0.156	&				&	139	\\
Wooden products for construction	&	1.953	&	***	&	-0.164	&	**		(**)	&	114	\\
Bolts, nuts, screws, rivets, and washers	&	1.688	&	**	&	-0.168	&	*			&	139	\\
Pig iron	&	1.922	&	***	&	-0.171	&				&	138	\\
Railroad passenger transport	&	2.544	&	***	&	-0.181	&	**		(*)	&	135	\\
Gold and silver ingots	&	2.860	&	***	&	-0.186	&			(*)	&	112	\\
Sand and gravel	&	2.520	&	**	&	-0.201	&				&	113	\\
Steel ships	&	1.549	&	**	&	-0.203	&				&	181	\\
Telecommunications	&	1.623	&	*	&	-0.213	&				&	123	\\
Other personal repair services	&	1.925	&	***	&	-0.225	&	***		(***)	&	147	\\
Education (public)	&	1.936	&	***	&	-0.231	&	***		(**)	&	169	\\
Gasoline and Jet oil	&	1.698	&	**	&	-0.234	&				&	127	\\
Other ships	&	1.888	&	***	&	-0.287	&	***		(***)	&	166	\\
Forgings	&	2.125	&	***	&	-0.289	&	***		(**)	&	122	\\
Cargo handling	&	1.861	&	**	&	-0.373	&			(***)	&	122	\\
Research and experiment in enterprise	&	1.415	&	**	&	-0.502	&			(***)	&	225	\\
Education (private, non-profit)	&	1.525	&	*	&	-0.509	&			(***)	&	148	\\
\hline
\end{xtabular}
\\\\
\normalsize
Note: The statistical significances in parenthesis are of the intercept of the regression (\ref{main}).
}

\end{document}